%% file: main.tex
\documentclass[a4paper,UKenglish,cleveref,autoref,english]{llncs}
\usepackage[utf8]{inputenc}
\usepackage{lineno}
\nolinenumbers

\usepackage{thmtools,thm-restate}
\usepackage{tikz}
\usepackage[hidelinks]{hyperref} 
\usepackage{chngcntr} 
\usepackage{stmaryrd}
\usepackage{latexsym}
\usepackage{amsmath, amssymb} 
\usepackage{mathtools}
\usepackage{logic9}
\usepackage{graphicx}
\usepackage{xparse}
\usepackage{suffix}
\usepackage{comment}

\usepackage[disable]{todonotes} 

\newcommand\withappendix 

\input{environments}
\input{macros}

\begin{document}
\title{One-way Resynchronizability \\of Word Transducers}
\author{Sougata Bose\inst{1} \and
	S.N.~Krishna\inst{2} \and
	Anca Muscholl\inst{1}\and
	Gabriele Puppis\inst{3}}
\authorrunning{S.~Bose, S.N.~Krishna, A.~Muscholl, G.~Puppis}
\institute{LaBRI, University of Bordeaux, France\and
	Dept.~of Computer Science \& Engineering IIT Bombay, India
	 \and
	Dept.~of Mathematics, Computer Science, and Physics, Udine University, Italy\\
}
\maketitle

\begin{abstract}
  The origin semantics for transducers was proposed in 2014, and it led to various characterizations and decidability results that are in contrast with the classical semantics. In this paper we add a further decidability result for characterizing transducers that are close to one-way transducers in the origin semantics. We show that it is decidable whether a non-deterministic two-way word transducer can be resynchronized by a bounded, regular resynchronizer into an origin-equivalent  one-way transducer. The result is in contrast with the usual semantics, where it is undecidable to know if a non-deterministic two-way transducer is equivalent to some one-way transducer.

\end{abstract}

	\input{intro}
	\input{prelims}
	\input{problem}
	\input{inversion}
	\input{proof-bounded}

\input{proof-unbounded}
    \input{conclusions}

	\bibliographystyle{abbrv}
	\bibliography{biblio,transducers}
	
	\longshort{
	\newpage
	\appendix
	\centerline{\bf \Large {Appendix}}

\counterwithin{theorem}{section}
\counterwithin{proposition}{section}
\counterwithin{lemma}{section}
\counterwithin{example}{section}
\counterwithin{remark}{section}

	\input{appendix-resynchs}
	\input{appendix-flows}
	\input{appendix-big}

\input{appendix-main}
	\input{appendix-unbounded.tex}
	}{
	}
\end{document}

%% file: environments.tex
\newcommand{\gabriele}[2][]{\todo[color=magenta!25,#1]{G. #2}\ignorespaces}

\usepackage{wrapfig}

\usetikzlibrary{automata,arrows,decorations.pathreplacing,patterns,positioning}
\usetikzlibrary{shapes.misc, positioning}
\pgfdeclarelayer{bg}    
\pgfsetlayers{bg,main}  
\tikzstyle{hidden}=[minimum size=1mm, shape=circle, inner sep=0pt, outer sep=0pt, style={font=\vphantom{Ag}}]
\tikzstyle{dot}=[draw,shape=circle, fill, minimum size=1mm, inner sep=0pt, outer sep=0pt]
\tikzstyle{state}=[draw,shape=circle, fill=gray!25, minimum size=6mm, inner sep=0pt, outer sep=0pt, style={font=\vphantom{Ag}}]
\tikzstyle{letter}=[draw,shape=rectangle, fill=gray!25, minimum size=5mm, inner sep=0pt, outer sep=0pt, style={font=\vphantom{Ag}}]
\tikzset{darrow/.style args={colored by #1 and #2}%
                             {line width=3pt,#1, 
                              postaction={draw,thin,#2}, 
                             }}
\tikzstyle{arrow} = [-, shorten >= 2pt, shorten <= 1pt, white, line width=2pt, 
                     postaction={draw, ->, >=stealth', line width=0.5pt,black}]

\tikzstyle{iletter}=[draw,shape=rectangle, minimum size=5mm, node distance=7mm, inner sep=0pt, outer sep=0pt, style={font=\vphantom{Ag}}]
\tikzstyle{oletter}=[draw,shape=rectangle, minimum size=5mm, node distance=7mm, inner sep=0pt, outer sep=0pt, style={font=\vphantom{Ag}}]

\newsavebox{\wraptext}
\newsavebox{\wrapfig}
\newsavebox{\deftext}
\savebox{\deftext}{Xp}
\newcommand\mywraptextjustified[2]{\savebox{\wraptext}{\parbox[t]{\linewidth-#1}%
                                                                 {#2\linebreak}}}
\newcommand\mywraptext[2]{\savebox{\wraptext}{\parbox[t]{\linewidth-#1}%
                                                        {#2\par\phantom{X}}}}
\newcommand\mywrapfig[2]{\savebox{\wrapfig}{\parbox[t]{#1}{\centering #2}}%
                         \usebox{\wraptext}%
                         \raisebox{-\ht\wrapfig}{\usebox{\wrapfig}}%
                         \vspace{-\ht\deftext minus 1pt}\linebreak}

\newcommand\longshort[2]{\ifdefined\withappendix #1\else #2\fi}

%% file: macros.tex
\newcommand\dom{\mathsf{dom}}
\newcommand{\lftmark}{\mathop{\vdash}}
\newcommand{\rgtmark}{\mathop{\dashv}}
\newcommand{\pspace}{\textsf{PSPACE}}

\newcommand{\pump}{\mathit{pump}}

\newcommand{\flow}{\mathit{flow}}
\NewDocumentCommand\flw{O{}m}{\mathit{flow}_{#1}(#2)}
\NewDocumentCommand\out{O{}m}{\mathit{out}_{#1}(#2)}
\NewDocumentCommand\bout{O{}m}{\mathit{bigout}_{#1}(#2)}
\newcommand{\nd}{p}

\newcommand\ipar{\mathsf{ipar}}
\newcommand\opar{\mathsf{opar}}
\newcommand\move{\mathsf{move}}
\newcommand\nxt{\mathsf{next}}

\definecolor{nicecyan}{HTML}{006165}
\definecolor{nicered}{HTML}{DB3A34}

\newcommand\otype{\tau}

\newcommand{\short}[1]{\mathit{low}(#1)}

\newcommand\oeq{=_o}

\newcommand{\orig}{\mathit{orig}}

\newcommand\oelement{\sigma}

\newcommand{\Qright}{Q_{\succ}}
\newcommand{\Qleft}{Q_{\prec}}
\newcommand{\lft}{\mathsf{L}}
\newcommand{\rgt}{\mathsf{R}}

\newcommand{\juxt}{\cdot}

\newcommand{\obound}{\mathbf{C}}
\newcommand{\fbound}{\mathbf{M}}

\newcommand{\an}[1]{\mathit{an}(#1)}

%% file: intro.tex
\section{Introduction}

Regular word-to-word functions form a robust and expressive class of
transformations, as they correspond to deterministic
two-way transducers,  to deterministic streaming string
transducers~\cite{AlurCerny10}, and to monadic second-order
logical transductions~\cite{eh01}. However, the transition from
word languages  to functions over words is often quite tricky.  One of
the challenges is to come up with effective
characterizations  of restricted transformations. A first  example is the
characterization of functions computed by one-way transducers 
(known as~\emph{rational functions}). It turns out that it is decidable
whether a regular function is rational~\cite{fgrs13}, but
the algorithm is quite involved~\cite{bgmp18}. In addition, non-determinism makes the
problem intractable: it is undecidable whether the relation computed
by a non-deterministic two-way transducer can be also computed  by a
one-way transducer, \cite{bgmp15}. A second example is the problem of
knowing whether a regular word function can be described by a
first-order logical transduction. This question is still open in
general~\cite{fkt14}, and it is only known how to decide if a \emph{rational}
function is definable in first-order logic~\cite{FGL19}.

Word transducers with origin semantics were introduced by
Boja\'nczyk~\cite{boj14icalp} and shown to provide a machine-independent
characterization of regular word-to-word functions.  The origin semantics, as the name suggests, means tagging the
output by the positions of the input that generated that output.

A nice phenomenon is that
origins can restore decidability for some interesting problems. For
example, the equivalence of word relations computed by
one-way transducers, which is undecidable in the classical
semantics~\cite{gri68,iba78siam}, is \pspace-complete for two-way non-deterministic
transducers  in the
origin semantics~\cite{bmpp18}. Another, deeper, observation is that the origin
semantics provides an algebraic
approach that can be used to decide fragments. For example,
\cite{boj14icalp} provides an effective characterization of
first-order definable word functions under the  origin semantics. As
for the problem of knowing whether a regular word function is
rational, it becomes almost trivial in the origin semantics.
 \begin{figure}[t]
	 	\begin{center}
	 		\begin{tikzpicture}[yscale=0.8]
	 		\node[iletter,draw=none] (i1) at (0,0) {$b$};
	 		\node[iletter,draw=none] (i2) [right of=i1] {$a$};
	 		\node[iletter,draw=none] (i3) [right of=i2] {$c$};
	 		\node[iletter,draw=none] (i4) [right of=i3] {$a$};
	 					
	 		\node[oletter,draw=none] (o1) at (0,-1.5) {$a$};
	 		\node[oletter,draw=none] (o2)  [right of=o1] {$b$};
	 		\node[oletter,draw=none] (o3) [right of=o2] {$a$};
	 		\node[oletter,draw=none] (o4) [right of=o3] {$c$};

	 		\node [anchor=east] (input1) at (-0.6,0) {input:};
	 		\node [anchor=east] (ouput1) at (-0.6,-1.5) {output:};
	 		\node [anchor=east] (origins1) at (-0.6,-0.75) {origins:};
	 		
		    \node[iletter,draw=none] (i5) at (7,0) {$b$};
	 		\node[iletter,draw=none] (i6) [right of=i5] {$a$};
	 		\node[iletter,draw=none] (i7) [right of=i6] {$c$};
	 		\node[iletter,draw=none] (i8) [right of=i7] {$a$};
			
		    \node[oletter,draw=none] (o5) at (7,-1.5)  {$a$};
	 		\node[oletter,draw=none] (o6) [right of=o5] {$b$};
	 		\node[oletter,draw=none] (o7) [right of=o6] {$a$};
	 		\node[oletter,draw=none] (o8) [right of=o7] {$c$};
			
	 		\node [anchor=east] (input2) at (7-0.6,0) {input:};
	 		\node [anchor=east] (ouput2) at (7-0.6,-1.5) {output:};
	 		\node [anchor=east] (origins2) at (7-0.6,-0.6) {resynchronized};
	 		\node [anchor=east] (origins2) at (7-0.6,-0.9) {origins:};
			
	 		\draw[arrow] (o1.north)--(i4.south);
	 		\draw[arrow] (o2.north)--(i1.south);
	 		\draw[arrow] (o3.north)--(i2.south);
	 		\draw[arrow] (o4.north)--(i3.south);
	 	
     	 	\draw[arrow] (o5.north)--([xshift=-0.3mm]i5.south);
	 		\draw[arrow] (o6.north)--([xshift=0.3mm]i5.south);
	 		\draw[arrow] (o7.north)--(i6.south);
	 		\draw[arrow] (o8.north)--(i7.south);
	 		 		\end{tikzpicture}
	 		 	\end{center}
	 
	 \caption{On the left, an input-output pair for a
           transducer $T$ that reads $wd$ and outputs $dw$, $d \in\S$,
           $w \in\S^*$, the arrows denoting origins. 
	 	 On the right, the same input-output pair, but with origins modified by a resynchronizer $\Rr$.   
	 	 The resynchronized relation $\Rr(T)$ is order-preserving, and $T$ is one-way resynchronizable.  
	 	  }
	 \label{fig:intro}
	 	 \end{figure}
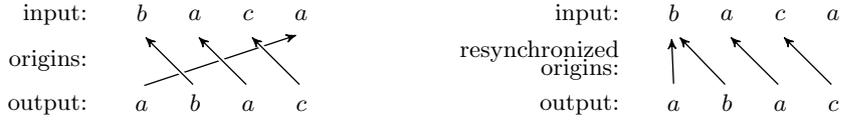

A possible objection against the origin semantics is that the
comparison of two transducers in the origin semantics is too
strict. Resynchronizations were proposed in order to
overcome this issue. A resynchronization  is
a binary relation between input-output pairs with origins, that
preserves the input and the output, changing only the origins. Resynchronizations were introduced for
one-way transducers~\cite{fjlw16icalp}, and later for two-way
transducers~\cite{bmpp18}.   For
one-way transducers \emph{rational} resynchronizations are
transducers acting on the synchronization languages, whereas for
two-way transducers, \emph{regular} resynchronizations are described by
regular properties over the input that restrict the change of
origins. The class of bounded\footnote{``Bounded'' refers here to the number
 of source positions that are mapped to the same target position. It
 rules out resynchronizations such as the universal one.} regular resynchronizations was shown to
behave very nicely, preserving the class of transductions
defined by non-deterministic, two-way transducers: for any bounded
regular resynchronization $\Rr$ and any two-way
transducer $T$, the resynchronized relation $\Rr(T)$ can be computed by
another two-way transducer~\cite{bmpp18}. In particular, non-deterministic, two-way
transducers can be effectively compared modulo bounded regular
resynchronizations.

As mentioned above, it is easy to know if a two-way transducer is
equivalent under the origin semantics to some one-way
transducer~\cite{boj14icalp}, since this is equivalent to being
order-preserving. But what happens if this is not
the case? Still, the given transducer $T$ can be
``close'' to some order-preserving  transducer. What we mean here by
``close'' is that there exists some bounded regular resynchronizer $\Rr$ such
that $\Rr(T)$ is order-preserving and all input-output pairs with
origins produced by $T$ are in the domain of $\Rr$. We call such
transducers \emph{one-way resynchronizable}. Figure \ref{fig:intro}
gives an example. 

In this paper we show that it is decidable if a two-way transducer is
one-way resynchronizable. We first solve the problem for bounded-visit
two-way transducers. A bounded-visit transducer is one for which there
is a uniform bound for the number of visits of any input
position. Then, we use the previous result to show that one-way
resynchronizability is decidable for arbitrary two-way transducers, so
without the bounded-visit restriction. This is done by constructing,
if possible, a bounded, regular resynchronization from the given
transducer to a bounded-visit transducer with regular language
outputs. Finally, we show that bounded regular
resynchronizations are closed under composition, and this allows to
combine the previous construction with our  decidability result
for bounded-visit transducers.

\medskip

\noindent
\emph{Related work and paper overview.} 
The synthesis problem for resynchronizers asks to
compute a resynchronizer from one transducer to another one, when
the two transducers are given as input. The problem was studied
in~\cite{bkmpp19} and  shown to be
decidable for unambiguous two-way transducers (it is open for
unrestricted transducers). The paper~\cite{KM20} shows that the
containment version of the above problem is
undecidable for unrestricted one-way transducers.

 The origin semantics for streaming string transducers (SST)
 \cite{AlurCerny10} 
 has been studied in \cite{BDGPicalp17}, providing a
 machine-independent characterization of  the sets of  origin graphs
 generated 
 by SSTs. An open problem here is to characterize origin
 graphs generated by aperiodic  
 streaming string transducers \cite{DBLP:journals/ijfcs/DartoisJR18,fkt14}.  
 Going beyond words, \cite{DBLP:journals/iandc/FiliotMRT18} investigates 
 decision problems of tree transducers with origin, and regains the  decidability 
 of the equivalence problem for non-deterministic top-down and MSO transducers by 
 considering the origin semantics. An open problem for tree transducers 
 with origin is that of synthesizing resynchronizers as in the word case.

 We will recall regular resynchronizations in Section~\ref{sec:resync}.
 Section~\ref{sec:defs} provides the proof ingredients for  the
 bounded-visit case, and the proof of decidability of one-way
 resynchronizability in the bounded-visit case can be found in Section
 \ref{sec:main}. Finally, in Section \ref{sec:general-case} we sketch the
 proof in the general case. 
\longshort{
Missing proofs can be found in the appendix.
}{
A full version of the paper is available at 
\url{https://arxiv.org/abs/????.????}.
}


%% file: prelims.tex
\section{Preliminaries}\label{sec:preliminaries}
                  
Let $\Sigma$ be a finite input alphabet. Given a word $w\in \Sigma^*$
of length $|w|=n$, a \emph{position} is an element 
of its domain $\dom(w)=\{1,\dots,n\}$. For every position $i$,
$w(i)$ denotes the letter at that position.
A \emph{cut} of $w$ is any number from $1$ to $|w|+1$, so a cut
identifies a position \emph{between} two consecutive
letters of the input. The cut $i=1$ represents
the position just before the first input letter, and  $i=|w|+1$ 
the position just after the last letter of $w$.

\paragraph{Two-way transducers.}

We use two-way transducers as defined in~\cite{bgmp18,bkmpp19},
with a  slightly different presentation than in classical papers such as~\cite{she59}.
As usual for two-way machines, for any input $w\in\Sigma^*$, 
$w(0)=\lftmark$ and $w(|w|+1)=\rgtmark$, where 
$\lftmark,\rgtmark\notin\S$ are special markers 
used as delimiters. 
A \emph{two-way transducer} (or just \emph{transducer} from now on)
is a tuple $T=(Q,\S,\G,\D,I,F)$, where 
$\S,\G$ are respectively the input and output alphabets,
$Q=\Qleft \uplus \Qright$ is the 
set of states, partitioned into left-reading states from $\Qleft$ 
and right-reading states from  $\Qright$, $I\subseteq \Qright$ is 
the set of initial states, $F\subseteq Q$ is the set of final states, and
$\Delta \subseteq Q\times (\Sigma\uplus\{\lftmark,\rgtmark\}) \times \G^* \times Q$ 
is the finite transition relation. 
Left-reading states read the letter 
to the left, whereas right-reading states read the letter to the right.
This partitioning will also determine the head movement during a transition,
as explained below.



As usual, to define runs of transducers we  first define configurations.
Given a transducer $T$ and a word $w\in\Sigma^*$,  
a \emph{configuration} of $T$ on $w$ is a state-cut pair $(q,i)$,
with $q\in Q$ and $1\le i\le |w|+1$. 
A configuration $(q,i)$, $1 \leq i \leq |w|+1$  means that the automaton 
is in state $q$ and its head is {\sl between} the $(i-1)$-th 
and the $i$-th letter of 
$w$.  
The transitions that depart from a configuration 
$(q,i)$ and read $a$ are denoted $(q,i) \act{a} (q',i')$, 
and must satisfy one of the following:\\
(1)  $q \in \Qright$, $q' \in \Qright$, $a=w(i)$, $(q,a,v,q') \in\D$, and $i'=i+1$, \\
(2)  $q \in \Qright$, $q' \in \Qleft$, $a=w(i)$, $(q,a,v,q') \in\D$, and $i'=i$, \\
(3) $q \in\Qleft$, $q' \in \Qright$, $a=w(i-1)$, $(q,a,v,q') \in\D$, and $i'=i$,\\ 
(4) $q \in\Qleft$, $q' \in \Qleft$, $a=w(i-1)$, $(q,a,v,q') \in\D$, and $i'=i-1$. 
When $T$ has only right-reading states (i.e.~$\Qleft=\emptyset$),
its head can only move rightward. In this case we call $T$ a 
\emph{one-way transducer}. 

A \emph{run} of $T$ on $w$ is a sequence 
$\r=(q_1,i_1) \act{a_{j_1}\mid v_1} (q_2,i_2) \act{a_{j_2} \mid v_2}
\cdots \act{a_{j_m}\mid v_m} (q_{m+1},i_{m+1})$ 
of configurations connected by transitions.
Note that the positions $j_1,j_2,\dots,j_m$ of letters 
do not need to be ordered from smaller to bigger,
and can differ slightly (by $+1$ or $-1$) 
from the cuts $i_1,i_2,\dots,i_{m+1}$, since cuts take values 
in between consecutive letters.

A configuration $(q,i)$ on $w$ is \emph{initial} (resp.~\emph{final})
if $q\in I$ and $i=1$ 
(resp.~$q\in F$ and $i=|w|+1$). 
A run is \emph{successful} if it starts with an initial configuration 
and ends with a final configuration. 
The \emph{output} associated with a successful run
$\rho$ as above 
is the word 
$v_1 v_2 \cdots v_m\in\Gamma^*$. 
A transducer $T$ defines a relation $\sem{T}\subseteq\S^*\times\Gamma^*$ 
consisting of all the pairs $(u,v)$ such that $v$ is the output of some successful 
run $\rho$ of $T$ on $u$.


\paragraph{Origin semantics.}
In the origin semantics for transducers \cite{boj14icalp} the output
is tagged with information about the position of the input where it
was produced.
If reading the $i$-th letter of the input we output $v$, 
then all letters of $v$ are tagged with $i$, and we say 
they have \emph{origin} $i$. 
We use the notation $(v,i)$ for $v \in\G^*$ to denote that all
positions in the output word $v$ have
origin $i$, and we view $(v,i)$ as word over the alphabet  $\G \times \Nat$.
The outputs associated with a  successful run 
$\rho= (q_1,i_1) \act{b_1 \mid v_1} (q_2,i_2) \act{b_2 \mid v_2}  
(q_3,i_3) \cdots \act{b_m \mid v_m} (q_{m+1},i_{m+1})$
in the origin semantics are the words of the form 
$\nu=(v_1,j_1) (v_2,j_2) \cdots (v_m,j_m)$ over $\G \times \Nat$
 where, for all $1\le k\le m$, $j_k=i_k$ if $q_k\in \Qright$, and 
 $j_k=i_k-1$ if $q_k\in \Qleft$.
Under the origin semantics, the relation defined by $T$, 
denoted $\sem{T}_o$, is the set of pairs $\oelement=(u,\nu)$
---called \emph{synchronized pairs}--- such that
$u\in\Sigma^*$ and $\nu\in (\G \times \Nat)^*$ 
is the output of some successful run on $u$. 

Equivalently,  a synchronized pair $(u,\nu)$ can be described as a
triple $(u,v,\orig)$, where $v$ is the projection of $\nu$ on $\G$,
and $\orig : \dom(v) \to \dom(u)$ associates with each position of $v$
its origin in $u$. So for  $\nu=(v_1,j_1) (v_2,j_2) \cdots (v_m,j_m)$
as above, $v=v_1 \dots v_m$, and, for all positions $i$ s.t. 
$|v_1 \dots v_{k-1}| < i \le |v_1 \dots v_k|$, we have 
$\orig(i)=j_k$.  
Given two transducers $T_1,T_2$, we say they are
\emph{origin-equivalent} if $\sem{T_1}_o=\sem{T_2}_o$.  Note
that two transducers $T_1,T_2$ can be equivalent in the classical
semantics, $\sem{T_1}=\sem{T_2}$, while they can have
different origin semantics, so $\sem{T_1}_o\neq\sem{T_2}_o$.

\paragraph{Bounded-visit transducers.} 
Let $k >0$ be some integer, and $\r$ some
run of a two-way transducer $T$. 
We say that $\r$ is \emph{$k$-visit} if for every $i \ge 0$,
it has at most $k$ occurrences of configurations from 
$Q \times \set{i}$.  
We call a transducer $T$ \emph{$k$-visit} if for every
$\s \in \sem{T}_o$ there is some successful, $k$-visit run
$\r$ of $T$ with output $\s$ (actually we should call the transducer
$k$-visit \emph{in the origin semantics}, but for simplicity we
omit this). 
For example,  the relation 
$\{(w, \overline{w}) \mid w \in \Sigma^*\}$, where $\overline{w}$ denotes 
the reverse of $w$, can be computed by a $3$-visit
transducer.
A transducer is called \emph{bounded-visit} if it is $k$-visit for some $k$. 

\paragraph{Common guess.} 
It is often useful to work with a variant of two-way transducers
that can guess beforehand some annotation on the input and inspect 
it consistently when visiting portions of the input multiple times. 
This feature is called \emph{common guess} \cite{BDGPicalp17},
and strictly increases the expressive power of two-way transducers, including bounded-visit
ones.


%% file: problem.tex
\section{One-way resynchronizability}
\label{sec:resync}


\subsection{Regular resynchronizers} \label{ssec:regresync}
Resynchronizations are used to compare transductions in the origin
semantics. A \emph{resynchronization}  is a binary relation
$\Rr \subseteq (\S^*\times(\G\times\Nat)^*)^2$ over synchronized pairs such that 
$(\oelement,\oelement') \in \Rr$ implies that $\oelement=(u, v, \orig)$ and 
$\oelement'=(u, v,\orig')$ for some origin mappings $\orig,\orig':\dom(v) \to
\dom(u)$. In other words, a resynchronization will only change the
origin mapping, but neither the input, nor the output.
Given a relation
$S \subseteq \S^*\times (\G\times\Nat)^*$ with origins, 
the \emph{resynchronized relation} $\Rr(S)$ is defined as
$ \Rr(S) = \{ \oelement'  \mid  (\oelement,\oelement')\in \Rr, \;
\oelement\in S\}$. For a transducer $T$ we  abbreviate
$\Rr(\sem{T}_o)$ by $\Rr(T)$. The typical use of a resynchronization $\Rr$ is to
ask, given two transducers $T,T'$, whether $\Rr(T)$ and $T'$ are origin-equivalent. 
%

Regular resynchronizers (originally called MSO resynchronizers)
were introduced in~\cite{bmpp18} as a resynchronization mechanism
that preserves definability by two-way transducers. They were
inspired by MSO (monadic second-order) transductions \cite{CE12,EH07} and they are formally
defined as follows.
A \emph{regular resynchronizer} is a tuple
$\Rr=(\bar I,\bar O,\ipar,\opar,(\move_\otype)_{\otype},(\nxt_{\otype,\otype'})_{\otype,\otype'})$ 
consisting of
\begin{itemize}
  \item some monadic parameters (colors) $\bar I = (I_1,\dots,I_m)$ and $\bar O = (O_1,\dots,O_n)$,
  \item MSO sentences $\ipar,\opar$, defining languages over expanded input and output alphabets,
        i.e.~over $\S'=\S\times2^{\{1,\dots,m\}}$ and $\G'=\G\times2^{\{1,\dots,n\}}$, respectively,
  \item MSO formulas $\move_\otype(y,z)$, $\nxt_{\otype,\otype'}(z,z')$ with two 
        free first-order variables and parametrized by expanded output
        letters $\otype,\otype'$ (called types, see below).
\end{itemize}
To apply a regular resynchronizer as above, one first guesses the
valuation of all the
predicates $I_j,O_k$, and uses it to interpret the 
parameters $\overline{I}$ and $\overline{O}$.
Based on the chosen valuation of the parameters $\overline{O}$, each position $x$ of 
the output $v$ gets an associated \emph{type} $\otype_x = (v(x), b_1, \dots, b_n) \in \G \times \{0,1\}^n$, 
where $b_j$ is $1$ or $0$ depending on whether $x \in O_j$ or not. We
refer to  the output word together with the valuation of the
output parameters as \emph{annotated output}, so a word over $\G \times \{0,1\}^n$. Similarly,
the \emph{annotated input} is a word over $\S \times \set{0,1}^m$. 
The annotated input and output word must
satisfy the formulas $\ipar$ and $\opar$, respectively.

The origins of output positions are constrained
using the formulas $\move_\otype$ and $\nxt_{\otype,\otype'}$, 
which are \emph{parametrized by output types and
evaluated over the annotated input}. Intuitively, the formula $\move_\otype(y,z)$ states how the origin of every output 
position of type $\otype$ changes from $y$ to $z$. We refer to 
$y$ and $z$ as \emph{source} and \emph{target} origin, respectively. 
The formula $\nxt_{\otype, \otype'}(z,z')$ instead constrains the target 
origins $z,z'$ of any two consecutive output positions with types $\otype$ and $\otype'$, respectively.

Formally, 
$\Rr=(\bar I,\bar O,\ipar,\opar,(\move_\otype),(\nxt_{\otype,\otype'}))$
defines the resynchronization consisting of all pairs 
$(\s,\s')$, with $\s=(u,v,\orig)$, $\s'=(u,v,\orig')$, 
$u \in \S^*$, and $v \in \G^*$, for which there exist $u' \in {\S'}^*$ and
$v' \in {\G'}^*$ such that
\begin{itemize}
  \item $\pi_\S(u')=u$ and $\pi_\G(v')=v$ 
  \item $u'$ satisfies $\ipar$ and $v'$ satisfies $\opar$,
  \item $(u',\orig(x),\orig'(x))$ satisfies $\move_\otype$ for all $\otype$-labeled output positions $x \in\dom(v')$, and
  \item $(u',\orig'(x), \orig'(x+1))$ satisfies $\nxt_{\otype,\otype'}$ 
        for all $x, x+1 \in \dom(v')$ such that $x$ and $x+1$ have label $\otype$ and $\otype'$, respectively. 
\end{itemize}
      
\begin{example}\label{ex:regular-resync}
  Consider the following resynchronization $\Rr$. 
  A pair $(\s,\s')$ belongs to $\Rr$ if 
  $\s=(uv,uwv,\orig)$, $\s'=(uv,uwv,\orig')$, 
  with $u,v,w \in\S^+$. 
  The origins $\orig$ and $\orig'$ are both the identity over $u$ and $v$. 
  The origin of every position of $w$ in $\s$ (hence a source origin)
  is either the first or the last position of $v$. 
  The origin of every position of $w$ in $\s'$ (a target origin) 
  is the first position of $v$.

  This resynchronization is described by a regular resynchronizer that
  uses two input parameters $I_1,I_2$ to mark the last and the first 
  positions of $v$ in the input, and one output parameter $O$ to mark 
  the factor $w$ in the output. 
  The formula $\move_{\otype}(y,z)$ is either $(I_1(y) \vee I_2(y)) \wedge I_2(z)$
  or $(y=z)$, depending on whether the type $\otype$ describes a position
  inside $w$ or a position outside $w$.
\end{example}

We now turn to describing some important restrictions on (regular) resynchronizers.
Let $\Rr=(\bar I,\bar O,\ipar,\opar,(\move_\otype),(\nxt_{\otype,\otype'}))$ be a resynchronizer.
\begin{itemize}
  \item $\Rr$ is \emph{$k$-bounded} (or just \emph{bounded}) 
        if for every annotated input $u' \in {\S'}^*$, 
        every output type $\otype \in \G'$, and every position $z$, 
        there are at most $k$ positions $y$ such that $(u',y,z)$ 
        satisfies $\move_\otype$. Recall that $y,z$ are input positions.


  \item $\Rr$ is \emph{$T$-preserving} for a given transducer $T$,
        if every $\s \in\sem{T}_o$ belongs to the domain of $\Rr$.
  \item $\Rr$ is \emph{partially bijective} if each $\move_\t$
    formula defines a
    partial, bijective function from source origins to target
    origins. Observe that this property implies that $\Rr$ is
    1-bounded. 
        

\end{itemize}

 The boundedness restriction rules out resynchronizations such as the
 universal one, that imposes no restriction on the change of origins. It is a decidable restriction~\cite{bmpp18}, and 
        it guarantees that definability by two-way transducers is effectively
        preserved under regular resynchronizations, modulo common guess.
        More precisely, Theorem 16 in~\cite{bmpp18} shows that, given 
        a bounded regular resynchronizer $\Rr$ and a transducer $T$, 
        one can construct a transducer $T'$ with common guess that is
        origin-equivalent to $\Rr(T)$.

\addtocounter{example}{-1}
\begin{example}[continued]
Consider again the regular resynchronizer 
$\Rr$ 
described in the previous example.
Note that $\Rr$ is $2$-bounded, since at most two source origins are redirected
to the same target origin. 
If we used an additional output parameter to distinguish, among the positions of $w$,
those that have source origin in the first position of $v$ and those that have source
origin in the last position of $v$, we would get a $1$-bounded, regular resynchronizer.
\end{example}


We state below two crucial properties of regular resynchronizers
(the second lemma is reminiscent of Lemma 11 from \cite{KM20},
which proves closure of bounded resynchronizers with vacuous 
$\nxt_{\otype,\otype'}$ relations).
	
\begin{restatable}{lemma}{OneBounded}\label{lem:one-bounded}
Every bounded, regular resynchronizer is effectively equivalent 
to some  $1$-bounded, regular resynchronizer.
\end{restatable}

\begin{restatable}{lemma}{ResyncClosure}\label{lem:resync-closure}
The class of bounded, regular resynchronizers is effectively
closed under composition.
\end{restatable}

\subsection{Main result} \label{ssec:mres}

Given a two-way transducer $T$ one can ask if it is origin-equivalent
to some one-way transducer. It was observed in~\cite{boj14icalp} that 
this property  holds if and only if all synchronized pairs 
defined by $T$ are \emph{order-preserving}, namely, for all 
$\s=(u,v,\orig)\in\sem{T}_o$ and all $y,y' \in \dom(v)$, with $y<y'$, 
we have $\orig(y) \le \orig(y')$.  
The decidability of the above question should be contrasted to
the analogous question in the classical semantics: 
``is a given  two-way transducer classically equivalent to some one-way transducer?'' 
The latter problem
turns out to be decidable for functional
transducers~\cite{fgrs13,bgmp18}, but is undecidable for arbitrary two-way
transducers~\cite{bgmp15}.

Here we are interested in a different, more relaxed notion:

\begin{definition}\label{def:oneway-resynchronizability}
A transducer $T$ is called \emph{one-way resynchronizable} if
there exists a bounded, regular resynchronizer $\Rr$ that is
$T$-preserving and such that $\Rr(T)$ is order-preserving.
\end{definition}

Note that if $T'$ is an order-preserving transducer, then one can
construct rather easily a one-way transducer $T''$ such that
$T'\oeq T''$, by eliminating non-productive U-turns from accepting
runs. 
 
Moreover, note that without the condition 
of being $T$-preserving 
every transducer $T$ would be
one-way resynchronizable, using the empty
resynchronization.

\begin{example}\label{ex:oneway-resynchronizability}
    Consider the transducer $T_1$ that moves the
  last letter of the input $wa$ to the front by a first left-to-right
  pass that outputs the last letter $a$, followed by a right-to-left
  pass without output, and finally by a left-to-right pass that
  produces the remaining $w$. Let $\Rr$ be the bounded regular
  resynchronizer that redirects the origin of the last $a$ to the
  first position. 
  Assuming an output parameter $O$ with an interpretation constrained 
  by $\opar$ that marks the last position of the output, the formula 
  $\move_{(a,1)}(y,z)$ says that target origin $z$ (source origin $y$,
  resp.) of the last $a$ is the first (last,
  resp.) position of the input. 
    It is easy to see that $\Rr(T_1)$ is origin-equivalent
  to the one-way transducer that on input $wa$, guesses $a$ and
  outputs $aw$. Thus, $T_1$ is one-way resynchronizable. See also Figure \ref{fig:intro}. 
  \label{eg1}
\end{example}
\begin{example}\label{ex:cross}
  Consider the transducer $T_2$ that reads inputs of the form $u \# v$
  and outputs $vu$ in the obvious way, by a first left-to-right pass
  that outputs $v$, followed by a right-to-left pass, and a finally a
  left-to-right pass that outputs $u$.
 %
  Using the characterization with the notion of cross-width that 
  we introduce below, it can be shown that $T_2$ is not one-way 
  resynchronizable.
\label{eg2}
\end{example}


In order to give a flavor of  our results, we anticipate here the two main theorems, 
before introducing the key technical concepts of cross-width and inversion 
(these will be defined further below).

\begin{theorem}\label{thm:oneway-def}
For every bounded-visit transducer $T$, the following are equivalent:
\begin{itemize}
  \item[(1)] $T$ is one-way resynchronizable,
  \item[(2)] the cross-width of  $T$ is finite,
  \item[(3)] no successful run of $T$ has inversions,
  \item[(4)] there is a 
    partially bijective, regular resynchronizer $\Rr$
             that is $T$-preserving and such that $\Rr(T)$ is
             order-preserving. 
\end{itemize}
Moreover, condition (3) is decidable.
\end{theorem}

We will use Theorem~\ref{thm:oneway-def} to show that one-way
resynchronizability is decidable for arbitrary two-way transducers
(not just bounded-visit ones). 

\begin{theorem}\label{thm:decidability}
It is decidable whether a given two-way transducer $T$ is one-way resynchronizable.
\end{theorem}

Let us now introduce the first key concept, that of cross-width:
\longshort{
  \begin{definition}[cross-width]\label{def:crosswidth}
Let $\s=(u,v,\orig)$ be  a synchronized pair and let 
$X_1,X_2 \subseteq \dom(v)$ be sets of output 
positions such that, for all $x_1\in X_1$ and 
$x_2\in X_2$, $x_1<x_2$ and $\orig(x_1)>\orig(x_2)$.
We call such a pair $(X_1,X_2)$ a \emph{cross} and define 
its \emph{width} as $\min(|\orig(X_1)|,|\orig(X_2)|)$, where
$\orig(X)=\{\orig(x) \mid x \in X\}$ is the set of origins 
corresponding to a set $X$ of output positions.

The \emph{cross-width} of a synchronized pair $\s$ is the maximal width 
of the crosses in $\s$. A transducer has \emph{bounded cross-width} if
for some integer $k$, all synchronized pairs associated with successful 
runs of $T$ have cross-width at most $k$.
\begin{center}
  \includegraphics[scale=0.58]{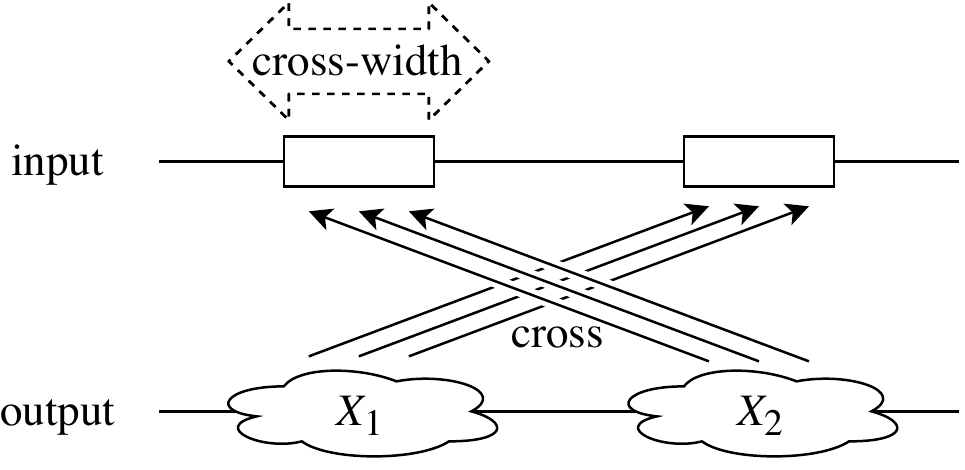}
\end{center}
\end{definition}
}
{
\begin{definition}[cross-width]\label{def:crosswidth}
Let $\s=$ \\
\mywraptextjustified{6.1cm}{
$(u,v,\orig)$ be  a synchronized pair and let 
$X_1,X_2 \subseteq \dom(v)$ be sets of output 
positions such that, for all $x_1\in X_1$ and 
$x_2\in X_2$, $x_1<x_2$ and $\orig(x_1)>\orig(x_2)$.
We call such a pair $(X_1,X_2)$ a \emph{cross} and define 
its \emph{width} as 
}
\mywrapfig{6.1cm}{
\vspace{-5.5mm}%
\hfill%
\includegraphics[scale=0.58,width=\linewidth,height=3cm]{figures/crosswidth.pdf}
}
$\min(|\orig(X_1)|,|\orig(X_2)|)$, where
$\orig(X)=\{\orig(x) \mid x \in X\}$ is the set of origins 
corresponding to a set $X$ of output positions. 
The \emph{cross-width} of a synchronized pair $\s$ is the maximal width 
of the crosses in $\s$. A transducer has \emph{bounded cross-width} if
for some integer $k$, all synchronized pairs associated with successful 
runs of $T$ have cross-width at most $k$.
\end{definition}
}

For instance, the transducer $T_2$ in Example~\ref{ex:cross} has unbounded
cross-width. In contrast, the transducer $T_1$ in
Example~\ref{ex:oneway-resynchronizability} has cross-width one.

The other key notion of \emph{inversion} will be introduced formally 
in the next section (page~\pageref{def:inversion}), as it requires 
a few technical definitions. The notion however is very similar in 
spirit to that of cross, with the difference that a single inversion 
is sufficient for witnessing a family of crosses with arbitrarily 
large cross-width.



%% file: inversion.tex
\section{Proof overview for Theorem~\ref{thm:oneway-def}}\label{sec:defs}

This section provides an overview of the proof of
Theorem~\ref{thm:oneway-def}, and introduces the main ingredients.

We will use flows (a concept inspired from crossing sequences \cite{she59,bgmp18}
     and revised in Section \ref{ssec:flows}) 
in order to derive the key notion of inversion. 
Roughly speaking, an inversion in a run involves two loops that produce 
outputs in an order that is reversed compared to the order on origins. 
Inversions were also used in the characterization of one-way definability
of two-way transducers under the classical
semantics~\cite{bgmp18}. There, they were used for
deriving some combinatorial properties of outputs. 
Here we are only interested in detecting inversions, 
and this is a simple task.

Flows will also be used to associate factorization trees with runs
(the existence of factorization trees of bounded height was established
by the celebrated Simon's factorization theorem~\cite{sim72}).
We will use a structural induction on these factorization trees
and the assumption that there is no inversion in every run to construct 
a regular resynchronization witnessing one-way resynchronizability
of the transducer at hand.

Another important ingredient underlying the main characterization is
given by the notion of dominant output interval (Section~\ref{ssec:big}), 
which is used to formalize the invariant of our inductive construction.

\subsection{Flows and inversions}\label{ssec:flows}
\paragraph{Intervals.}
An \emph{interval} of a word
is a set of consecutive positions in it.
An interval is often denoted by 
$I=[i,i')$, with $i=\min(I)$ and $i'=\max(I)+1$.
Given two intervals $I=[i,i')$ and $J=[j,j')$, we write 
$I<J$ if $i'\le j$, and we say that $I,J$ are adjacent 
if $i'=j$.
The union of two adjacent intervals $I=[i,i')$, $J=[j,j')$, 
denoted $I \juxt J$, is the interval $[i,j')$ (if $I,J$
are not adjacent, then $I\juxt J$ is undefined).

\paragraph{Subruns.}
Given a run $\rho$ of a transducer, 
a \emph{subrun} is a factor of $\rho$. Note that 
a subrun of a two-way transducer may visit a position of the input several
times. For an input interval $I=[i,j)$ and a run $\r$, we say that
a subrun $\r'$ of $\r$ \emph{spans over $I$} if $i$ (resp.~$j$)  is
the smallest (resp.~greatest) 
input position labeling some transition of $\r'$. 
The left hand-side of 
the figure at page \pageref{fig:pumping} gives an example 
of an interval $I$ of an input word together with the subruns
$\alpha_1, \alpha_2, \alpha_3, \beta_1, \beta_2, \beta_3$, $\gamma_1$
that span over it.
Subruns spanning over an interval can be left-to-right, left-to-left, 
right-to-left, or right-to-right depending on where the starting and ending 
positions are w.r.t.~the endpoints of the interval.


\paragraph{Flows.}
Flows are used to 
summarize subruns of a two-way transducer that span over a given
interval. 
The definition below is essentially taken from~\cite{bgmp18}, 
except for replacing ``functional'' by ``$K$-visit''.
Formally, a \emph{flow} of a transducer $T$ 
is a graph with vertices divided into two groups, 
$\lft$-vertices and $\rgt$-vertices, labeled by states of $T$, 
and with directed edges also divided into two groups, productive 
and non-productive edges. 
The graph satisfies the following requirements. Edge sources are
either an $\lft$-vertex
labeled by a right-reading state, or an $\rgt$-vertex 
labeled by a left-reading state, and symmetrically for edge
destinations; moreover, edges are of one of the following types:
$\lft\lft$, $\lft\rgt$, $\rgt\lft$, $\rgt\rgt$.
Second, each node is the endpoint of exactly one edge. 
Finally, $\lft$ ($\rgt$, resp.) vertices are 
totally ordered, in such a way that for every $\lft\lft$ ($\rgt\rgt$, resp.)
edge $(v,v')$, we have $v<v'$. We will  only  consider flows of
$K$-visiting transducers, so flows with at most $2K$ vertices. For
example, the flow in the left-hand side of the figure at page
\pageref{fig:pumping} has six $\lft$-vertices on the left, and six
$\rgt$-vertices on the right. The edges $\a_1$, $\a_2$, $\a_3$ are
$\lft\lft$, $\lft\rgt$, and $\rgt\rgt$, respectively.

Given a run $\r$ of $T$ and an interval $I=[i,i')$ on the input,
\emph{the flow of $\r$ on $I$}, denoted $\flw[\rho]{I}$, is
obtained by identifying every configuration at position $i$ 
(resp.~$i'$) with an $\lft$ (resp.~$\rgt$) vertex, labeled by the 
state of the configuration, and every subrun spanning over $I$ 
with an edge connecting the appropriate vertices (this subrun 
is called the \emph{witnessing subrun} of the edge of the flow).
An edge is said to be \emph{productive} 
if its witnessing subrun produces non-empty output.

\paragraph{Flow monoid.}
The composition of two flows $F$ and $G$ is defined when 
the $\rgt$-vertices of $F$ induce the same 
sequence of labels as the $\lft$-vertices of $G$. 
In this case, the composition results in the flow $F\cdot G$ that has as
vertices the $\lft$-vertices of $F$ and the $\rgt$-vertices of $G$, and for edges 
the directed paths in the graph obtained by glueing the $\rgt$-vertices of $F$ with 
the $\lft$-vertices of $G$ so that states are matched. 
Productiveness of edges is inherited by paths, implying that an edge 
of $F\cdot G$ is productive if and only if the corresponding path contains
at least one edge (from $F$ or $G$) that is productive.
When the composition is undefined, we simply write $F\cdot G=\bot$.
The above definitions naturally give rise to a \emph{flow monoid}
associated with the transducer $T$,
where elements are the flows of $T$, extended with a dummy
element $\bot$, and the product operation is given by
the composition of flows, with the convention that $\bot$ is
absorbing. 
It is  easy to verify that for any two adjacent intervals $I<J$
of a run $\r$, $\flw[\r]{I}\cdot\flw[\r]{J} = \flw[\r]{I\cdot J}$. 
We denote by $M_T$ the \emph{flow monoid} of 
a $K$-visiting transducer $T$. 

Let us estimate the size of $M_T$. If $Q$ is the set of states of $T$, there are
at most $|Q|^{2K}$ possible sequences of $\lft$ and $\rgt$-vertices;
and the number of edges (marked as productive or not) is bounded by 
${2K \choose K} \cdot (2K)^K \cdot 2^K \le (2K+1)^{2K}$. 
Including the dummy element $\bot$ in the flow monoid, we get
 $|M_T| \le (|Q| \cdot (2K+1))^{2K} +1 =: \fbound$.

\paragraph{Loops.} 
A loop of a run $\rho$ over input $w$ is an interval $I=[i,j)$ 
with a flow $F=$ 
\mywraptext{7.8cm}{
$\flw[\rho]{I}$ such that 
$F \cdot F =F$ (call $F$ \emph{idempotent}).
The run $\rho$ can be pumped
on a loop $I=[i,j)$ as expected: given $n>0$, 
we let $\pump^n_I(\r)$ be the run obtained from $\r$
by glueing the subruns that span over
the intervals $[1,i)$ and $[j,|w|+1)$ 
with $n$ copies of the subruns spanning over $I$ 
(see figure to the right). 
}
\mywrapfig{7.8cm}{
\vspace{-1.8mm}%
\hfill\scalebox{0.7}{\input{figures/pumping}}%
\label{fig:pumping}%
\label{l:pumping}%
}

The lemma below shows that the occurrence
order relative to subruns witnessing $\lft\rgt$ or $\rgt\lft$ edges of
a loop
(called~\emph{straight edges}, for short)
is preserved when pumping the loop.
This seemingly straightforward lemma is needed for
detecting inversions and its proof is surprisingly 
non-trivial.  
For example, the external edge connecting the two $\lft$-vertices
$1,2$ in the figure above appears before edge $\a_2$, and also before
every copy of $\a_2$ in the run where loop $I$ is pumped.

\begin{restatable}{lemma}{OrderPreservation}\label{lem:order-preservation}
Let $\r$ be a run of $T$ on $u$, let $J<I<K$ be a partition of the domain of $u$
into intervals, with $I$ loop of $\r$, 
and let $F=\flw[\r]{J}$, $E=\flw[\r]{I}$, and $G=\flw[\r]{K}$ 
be the corresponding flows.
Consider an arbitrary edge $f$ of either $F$ or $G$, and a straight edge $e$ of 
the idempotent flow $E$. Let $\r_f$ and $\r_e$ be the witnessing
subruns of $f$ and $e$, respectively.
Then the 
occurrence order of $\r_f$ and $\r_e$ in $\r$ is 
the same as the 
occurrence order of $\r_f$ and any copy of $\r_e$
in $\pump^n_I(\r)$. 
\end{restatable}

\medskip
We can now formalize the key notion of inversion:

\begin{definition}[inversion]\label{def:inversion}
An \emph{inversion} of $\r$ is a tuple $(I,e,I',e')$ such that
\mywraptext{6.8cm}{
\vspace{-3mm}%
\begin{itemize}
  \item $I,I'$ are loops of $\r$ and $I < I'$, 
  \item $e,e'$ are productive straight edges in  
        $\flw[\r]{I}$ and $\flw[\r]{I'}$ respectively,
  \item the subrun witnessing $e'$ precedes the
        subrun witnessing $e$ in the run order
\end{itemize}
(see the figure to the right).
}
\mywrapfig{6.8cm}{
\vspace{0mm}%
\hfill%
\scalebox{0.76}{\input{figures/inv}}%
}
\end{definition}


\subsection{Dominant output intervals}\label{ssec:big}
In this section we identify some particular intervals of the output 
that play an important role in the inductive construction of the 
resynchronizer for a one-way resynchronizable transducer.

Given $n\in\Nat$, we say that a set $B$ of output positions
is \emph{$n$-large} if $|\orig(B)| > n$; otherwise, we say that $B$ is
\emph{$n$-small}. Recall that here we work with a $K$-visiting
transducer $T$, for some constant $K$, and that $\fbound=(|Q| \cdot (2K+1))^{2K}+1$ 
is an upper bound to the size of the flow monoid $M_T$.
We will extensively use the derived constant 
$\obound=\fbound^{2K}$   
to distinguish
between large and small sets of output positions.
The intuition behind this constant is that any 
set of output positions that is $\obound$-large must 
traverse a loop of $\r$. 
This is formalized in the lemma below. The proof uses algebraic properties of 
 the flow monoid $M_T$ \cite{Ismael-bound}
(see also Theorem 7.2 in \cite{bgmp18}, which proves a 
similar result, but with a larger constant derived from
Simon's factorization theorem):

\begin{restatable}{lemma}{Bound}\label{lem:bound}
Let $I$ be an input interval and 
$B$ a set of output positions
with origins inside $I$.
If $B$ is $\obound$-large, then there is a loop $J\subseteq I$
of $\r$ such that $\flw[\r]{J}$ contains a productive straight edge
witnessed by a subrun that intersects $B$ (in particular, 
$\out{J}\cap B \neq \emptyset$).
\end{restatable}

We need some more notations for outputs.
Given an input interval $I$  we denote by 
$\out[\r]{I}$ 
the set of output positions whose origins belong to $I$ (note that 
this might not be an output interval). 
An \emph{output block} of $I$ is a maximal interval contained in $\out[\r]{I}$.

The \emph{dominant output interval} of $I$, denoted $\bout[\r]{I}$,
is the smallest output interval that contains all $\obound$-large output
blocks of $I$. In particular, $\bout[\r]{I}$ either is empty or begins 
with the first $\obound$-large output block of $I$ and ends with
the last $\obound$-large outblock block of $I$.
We will often omit the subscript $\r$ 
from the notations $\flw[\r]{I}$, $\out[\r]{I},\bout[\r]{I}$, etc., when no confusion arises.
\label{p:bigout}

We now fix a successful run $\r$ of the $K$-visiting transducer $T$. 
The rest of the section presents some technical lemmas that will be
used in the inductive constructions for the proof of the main theorem. 
\emph{In the lemmas below, we assume that all successful runs of $T$ 
(in particular, $\r$) avoid inversions.}

\begin{restatable}{lemma}{Order}\label{lem:order}
Let $I_1 < I_2$ be two input intervals and $B_1,B_2$ 
output blocks of $I_1$, $I_2$, respectively.
If both $B_1,B_2$ are $\obound$-large, then $B_1<B_2$. 
\end{restatable}

\begin{proof}[sketch]
If the claim would not hold, then Lemma \ref{lem:bound} would provide
some loops $J_1\subseteq I_1$ and $J_2\subseteq I_2$, together with  some
productive edges in them, witnessing an inversion.
\qed
\end{proof}

\begin{restatable}{lemma}{Bin}\label{lem:bin}
Let $I=I_1\juxt I_2$, $B=\bout{I}$, and $B_i=\bout{I_i}$ for $i=1,2$. 
Then $B\setminus(B_1\cup B_2)$ is $4K\obound$-small.
\end{restatable}
      
\begin{proof}[sketch]
By Lemma \ref{lem:order}, $B_1 < B_2$.
Moreover, all $\obound$-large output blocks of $I_1$ or $I_2$ 
are also $\obound$-large output blocks of $I$, so $B$ contains both $B_1$ and $B_2$.
Suppose, by way of contradiction, that $B\setminus(B_1\cup B_2)$ is
$4K\obound$-large. 
This means that there is a $2K\obound$-large set $S\subseteq B\setminus(B_1\cup B_2)$ 
with origins entirely to the left of $I_2$, or entirely to the right
of $I_1$.
Suppose, w.l.o.g., that the former case holds, and decompose $S$ as a union of 
maximal output blocks $B'_1,B'_2,\dots,B'_n$ with origins either
entirely inside $I_1$, or entirely outside.
Since $S\cap B_1 = \emptyset$,  every block $B'_i$ 
with origins inside $I_1$ is $\obound$-small.
\longshort{
Similarly, by Lemma \ref{lem:between} in Appendix~\ref{appendix:B}, 
}{
Similarly, one can prove that
}
every block $B'_i$ with origins outside $I_1$ is $\obound$-small too.
Moreover, since $\r$ is $K$-visiting, we get
$n \le 2K$.
Altogether, this contradicts the assumption that $S$ is $2K\obound$-large.
\qed
\end{proof}

\begin{restatable}{lemma}{Idem}\label{lem:idem}
Let $I = I_1\juxt I_2\cdots I_n$, such that $I$ is a loop
and $\flw{I}=\flw{I_k}$ for all $k$.  
Then $\bout{I}$ can be decomposed as $B_1 \juxt J_1 \juxt B_2 \juxt J_2 \juxt \dots \juxt J_{n-1} \juxt B_n$, 
where
\begin{enumerate}
  \item for all $1 \le k \le n$, $B_k=\bout{I_k}$ (with $B_k$ possibly empty);
  \item for all $1 \le k <n$, the positions in $J_k$ have origins 
		inside $I_k \cup I_{k+1}$ and 
		$J_k$ is $2K\obound$-small. 
\end{enumerate}
\end{restatable}

\begin{proof}[sketch]
The proof idea is similar to the previous lemma. First, using
properties of idempotent flows, one shows that 
all output positions strictly between $B_k$ and $B_{k+1}$,
for any $k=1,\dots,n-1$, have origin in $I_k\cup I_{k+1}$.
Then, one observes that every output block of $I_k$ disjoint
from $B_k$ is $\obound$-small, and since 
\linebreak
\mywraptext{5.1cm}{
$T$ is $K$-visiting
there are at most $K$ such blocks. This shows that every output 
interval $J_k$ between $B_k$ and $B_{k+1}$ is
$2K\obound$-small. 
For an illustration see the figure 
to the right. 
The $\obound$-large blocks in $I_1$ are shown in red; in blue those for $I_2$,
in purple those for $I_3$. So $\bout{I_1}$
is the entire output between the two red dots, 
$\bout{I_2}$ between the two blue dots, and $\bout{I_3}$ between the
purple dots. All three blocks are non-empty, and 
$\bout{I_1 \cdot I_2 \cdot I_3}$ goes from the first red to 
the second purple dot. Black non-dashed arrows stand for $\obound$-small blocks.
\qed
}
\mywrapfig{5.1cm}{
\vspace{0mm}%
\hfill%
\scalebox{0.65}{\input{figures/decomp-idemp}}%
}
\end{proof}


%% file: figures/pumping.tex
\begin{tikzpicture}[baseline=0, inner sep=0, outer sep=0, minimum size=0pt]
  \tikzstyle{dot} = [draw, circle, fill=white, minimum size=4pt]
  \tikzstyle{fulldot} = [draw, circle, fill=black, minimum size=4pt]
  \tikzstyle{grayfactor} = [->, shorten >=1pt, rounded corners=5, dashed]
  \tikzstyle{factor} = [->, shorten >=1pt, rounded corners=5]
  \tikzstyle{dotfactor} = [->, shorten >=1pt, dotted, rounded corners=5]
  \tikzstyle{fullfactor} = [->, >=stealth, shorten >=1pt, very thick, rounded corners=5]
  \tikzstyle{dotfullfactor} = [->, >=stealth, shorten >=1pt, dotted, very thick, rounded corners=5]

\begin{scope}[xscale=0.4, yscale=0.49]
\begin{pgfonlayer}{bg}
  \fill [pattern=north east lines, pattern color=gray!25]
        (4,-1) rectangle (8,11);
  \draw [dashed, thin, gray] (4,-1) -- (4,11);
  \draw [dashed, thin, gray] (8,-1) -- (8,11);
  \draw [gray] (4,-1.25) -- (4,-1.5) -- (8,-1.5) -- (8,-1.25);
  \draw [gray] (6,-1.75) node [below] {\footnotesize $I$};
\end{pgfonlayer}

  \draw (1.5,0) node (node0) {};
  
  \draw (3.4,0) node [circle, fill=white] (node01) {\small{0}};
  \draw (3.4,1) node [circle, fill=white] (node11) {\small{1}};
  \draw (3.4,2) node [circle, fill=white] (node12) {\small{2}};
  \draw (3.4,7) node [circle, fill=white] (node13) {\small{3}};
  \draw (3.4,8) node [circle, fill=white] (node14) {\small{4}};
  \draw (3.4,9) node [circle, fill=white] (node15) {\small{5}};
  \draw (3.4,10) node [circle, fill=white] (node16) {\small{6}};

  \draw (8.5,2) node [circle, fill=white] (node10) {\small{0}};
  \draw (8.5,3) node [circle, fill=white] (node11) {\small{1}};
  \draw (8.5,4) node [circle, fill=white] (node12) {\small{2}};
  \draw (8.5,5) node [circle, fill=white] (node13) {\small{3}};
  \draw (8.5,6) node [circle, fill=white] (node14) {\small{4}};
  \draw (8.5,7) node [circle, fill=white] (node15) {\small{5}};
  \draw (8.5,10) node [circle, fill=white] (node16) {\small{6}};

    \draw (4,0) node [dot, draw=blue!50] (node1) {};
  \draw (6,0) node (node2) {};
  \draw (6,1) node (node3) {};
  \draw (4,1) node [dot, draw=blue!50] (node4) {};
  \draw (2,1) node (node5) {};
  \draw (2,2) node (node6) {};
  \draw (4,2) node [fulldot, blue!50] (node7) {};
  \draw (8,2) node [dot, draw=blue!50] (node8) {};
  \draw (10,2) node (node9) {};
  \draw (10,3) node (node10) {};
  \draw (8,3) node [dot, draw=blue!50] (node11) {};
  \draw (6,3) node  (node12) {};
  \draw (6,4) node (node13) {};
  \draw (2,4) node (node14) {};
  \draw (4,4) node  (node15) {};
  \draw (8,4) node [dot, draw=blue!50] (node16) {};
  \draw (10,4) node (node17) {};
  
  \draw (10,5) node (node18) {};
  \draw (8,5) node [dot, draw=nicered] (node19) {};
  \draw (6,5) node (node20) {};
  \draw (6,6) node (node21) {};
  \draw (8,6) node [dot, draw=nicered] (node22) {};
  \draw (10,6) node (node23) {};
  \draw (10,6) node (node23) {};
  \draw (10,7) node (node24) {};
  \draw (8,7) node [fulldot, nicered] (node25) {};
  \draw (4,7) node [dot, draw=nicered](node26) {};
  \draw (2,7) node (node27) {};
  \draw (2,8) node (node28) {};
  \draw (4,8) node [dot, draw=nicered](node29) {};
  \draw (6,8) node (node30) {};
  \draw (6,9) node (node31) {};
  \draw (4,9) node [dot, draw=nicered] (node32) {};
  \draw (2,9) node (node33) {};

  \draw (2,10) node (node34) {};
  \draw (4,10) node [fulldot,gray](node35) {};
  \draw (8,10) node [dot, draw=gray](node36) {};
  \draw (10.5,10) node (node37) {};
  \draw (6,8) node (node38) {};
  \draw (6,8) node (node39) {};
  \draw (6,8) node (node40) {};

\begin{pgfonlayer}{bg}
  \draw [grayfactor,shorten <=10pt,shorten >=10pt] (node0) -- (node1);
  \draw [fullfactor, blue!50] (node1) -- (node2.center) -- node [below right=0.5mm] {\footnotesize $\pmb{\alpha_1}$}
                           (node3.center) -- (node4); 
  \draw [grayfactor,shorten <=10pt,shorten >=10pt] (node4) -- (node5.center) -- (node6.center) -- (node7); 
  \draw [fullfactor, blue!50] (node7) -- node [below=.7mm] {\footnotesize $~\pmb{\alpha_2}$} (node8);
  \draw [grayfactor,shorten <=10pt,shorten >=10pt] (node8) -- (node9.center) -- (node10.center) -- (node11);
  \draw [fullfactor, blue!50] (node11) -- (node12.center) -- node [below left=0.5mm] {\footnotesize $\pmb{\alpha_3}$} 
                           (node13.center) -- (node16);

  \draw [grayfactor,shorten <=10pt,shorten >=10pt] (node16) -- (node17.center) -- (node18.center) -- (node19);
  \draw [fullfactor, nicered] (node19) -- (node20.center) -- node [below left=0mm] {\footnotesize $\pmb{\beta_1}$}
                          (node21.center) -- (node22);
  \draw [grayfactor,shorten <=10pt,shorten >=10pt] (node22) -- (node23.center) -- (node24.center) -- (node25);
  \draw [fullfactor, nicered] (node25) -- node [below left=1.4mm] {\footnotesize $\pmb{\beta_2}$} (node26) ;
  \draw [grayfactor,shorten <=10pt,shorten >=10pt] (node26) -- (node27.center) -- (node28.center) -- (node29);
  \draw [fullfactor, nicered] (node29) -- (node30.center) -- node [below right=0.5mm] {\footnotesize $\pmb{\beta_3}$}
                           (node31.center) -- (node32);
  \draw [grayfactor,shorten <=10pt,shorten >=10pt] (node32) -- (node33.center) -- (node34.center) -- (node35);

  \draw [fullfactor, gray] (node35) -- node [below=0.5mm] {\footnotesize $~\pmb{\gamma_1}$} (node36);
  \draw [grayfactor,shorten <=10pt,shorten >=10pt] (node36) -- (node37);
\end{pgfonlayer}
\end{scope}

\begin{scope}[xshift=4cm, yshift=-0.3cm, xscale=0.4, yscale=0.3]
  \tikzstyle{factor} = [->, shorten >=1pt, dotted, rounded corners=4]
  \tikzstyle{fullfactor} = [->, >=stealth, shorten >=1pt, very thick, rounded corners=4]

  \fill [pattern=north east lines, pattern color=gray!25]
        (4,-1) rectangle (8,19);
  \fill [pattern=north east lines, pattern color=gray!25]
        (8,1) rectangle (12,19);
  \fill [pattern=north east lines, pattern color=gray!25]
        (12,3) rectangle (16,19);
  \draw [dashed, thin, gray] (4,-1) -- (4,19);
  \draw [dashed, thin, gray] (8,-1) -- (8,19);
  \draw [dashed, thin, gray] (12,-1) -- (12,19);
  \draw [dashed, thin, gray] (16,-1) -- (16,19);
  \draw [gray] (4,-1.25) -- (4,-1.5) -- (7.9,-1.5) -- (7.9,-1.25);
  \draw [gray] (8.1,-1.25) -- (8.1,-1.5) -- (12,-1.5) -- (12,-1.25);
  \draw [gray] (12.1,-1.25) -- (12.1,-1.5) -- (15.9,-1.5) -- (15.9,-1.25);
  \draw [gray] (6,-1.75) node [below] {\footnotesize $I$};
  \draw [gray] (12,-1.75) node [below] {\footnotesize $2$ more copies of $I$};

  \draw (1.5,0) node (node0) {};
  \draw (4,0) node [dot, draw=blue!50] (node1) {};
  \draw (6,0) node (node2) {};
  \draw (6,1) node (node3) {};
  \draw (4,1) node [dot, draw=blue!50] (node4) {};
  \draw (2,1) node (node5) {};
  \draw (2,2) node (node6) {};
  \draw (4,2) node [fulldot, blue!50] (node7) {};
  \draw (8,2) node [dot, draw=blue!50] (node8) {};
  \draw (10,2) node (node9) {};
  \draw (10,3) node (node10) {};
  \draw (8,3) node [dot, draw=blue!50] (node11) {};
  \draw (6,3) node  (node12) {};
  \draw (6,4) node (node13) {};
  \draw (8,4) node [fulldot, blue!50] (node14) {};
  \draw (12,4) node [dot, draw=blue!50] (node15) {};
  \draw (14,4) node  (node16) {};
  \draw (14,5) node  (node17) {};
  \draw (12,5) node [dot, draw=blue!50] (node18) {};
  \draw (10,5) node (node19) {};
  \draw (10,6) node  (node20) {};
  \draw (12,6) node [fulldot, blue!50] (node21) {};
  \draw (16,6) node [dot, draw=blue!50] (node22) {};
  \draw (18,6) node (node23) {};
  \draw (18,7) node  (node24) {};
  \draw (16,7) node [dot, draw=blue!50] (node25) {};
  \draw (14,7) node  (node26) {};
  \draw (14,8) node (node27) {};
  \draw (16,8) node [dot, draw=blue!50] (node28) {};
  \draw (18,8) node  (node29) {};
  
  \draw (18,9) node  (node30) {};
  \draw (16,9) node [dot, draw=nicered] (node31) {};
  \draw (14,9) node  (node311) {};
  \draw (14,10) node  (node312) {};
  \draw (16,10) node [dot, draw=nicered] (node313) {};
  \draw (18,10) node  (node314) {};
  \draw (18,11) node  (node315) {};
  \draw (16,11) node [fulldot, nicered] (node316) {};
  \draw (12,11) node [dot, draw=nicered] (node32) {};
  \draw (10,11) node  (node33) {};
  \draw (10,12) node  (node34) {};
  \draw (12,12) node [dot, draw=nicered] (node35) {};
  \draw (14,12) node (node36) {};
  \draw (14,13) node  (node37) {};
  \draw (12,13) node [fulldot, nicered] (node38) {};
  \draw (8,13) node [dot, draw=nicered] (node39) {};
  \draw (6,13) node  (node40) {};
  \draw (6,14) node (node41) {};
  \draw (8,14) node [dot, draw=nicered] (node42) {};
  \draw (10,14) node (node43) {};
  \draw (10,15) node (node44) {};
  \draw (8,15) node [fulldot, nicered] (node45) {};
  \draw (4,15) node [dot, draw=nicered] (node46) {};
  \draw (2,15) node (node47) {};
  \draw (2,16) node (node48) {};
  \draw (4,16) node [dot, draw=nicered] (node49) {};
  \draw (6,16) node  (node491) {};
  \draw (6,17) node  (node492) {};
  \draw (4,17) node [dot, draw=nicered] (node493) {};
  \draw (2,17) node  (node494) {};
  
  \draw (2,18) node  (node495) {};
  \draw (4,18) node [fulldot, nicecyan] (node496) {};  
  \draw (8,18) node [fulldot, gray] (node50) {};
  \draw (12,18) node [fulldot, gray] (node51) {};
  \draw (16,18) node [dot, draw=gray] (node52) {};
  \draw (18,18) node  (node53) {};

  \draw [grayfactor] (node0) -- (node1);
  \draw [fullfactor, blue!50] (node1) -- (node2.center) -- node [below right=0.5mm] {\footnotesize $\pmb{\alpha_1}$}
                           (node3.center) -- (node4); 
  \draw [grayfactor] (node4) -- (node5.center) -- (node6.center) -- (node7); 
  \draw [fullfactor, blue!50] (node7) -- node [below=.7mm] {\footnotesize $~\pmb{\alpha_2}$} (node8);
  \draw [fullfactor, blue!50] (node8) -- (node9.center) -- node [below right=0.5mm] {\footnotesize $\pmb{\alpha_1}$}
                           (node10.center) -- (node11);
  \draw [fullfactor, blue!50] (node11) -- (node12.center) -- node [below left=0.5mm] {\footnotesize $\pmb{\alpha_3}$}
                           (node13.center) -- (node14);
  \draw [fullfactor, blue!50] (node14) -- node [below=.7mm] {\footnotesize $~\pmb{\alpha_2}$} (node15);                 
  \draw [fullfactor, blue!50] (node15) -- (node16.center) -- node [below right=0.5mm] {\footnotesize $\pmb{\alpha_1}$}
                           (node17.center) -- (node18);
  \draw [fullfactor, blue!50] (node18) -- (node19.center) -- node [below left=0.5mm] {\footnotesize $\pmb{\alpha_3}$}
                           (node20.center) -- (node21);
  \draw [fullfactor, blue!50] (node21) -- node [below=.7mm] {\footnotesize $~\pmb{\alpha_2}$} (node22);
  \draw [grayfactor] (node22) -- (node23.center) -- (node24.center) -- (node25); 
  \draw [fullfactor, blue!50] (node25) -- (node26.center) -- node [below left=0.5mm] {\footnotesize $\pmb{\alpha_3}$}
                           (node27.center) -- (node28);
                     
  \draw [grayfactor] (node28) -- (node29.center) -- (node30.center) -- (node31); 
  \draw [fullfactor, nicered] (node31) -- (node311.center) -- node [below left=0mm] {\footnotesize $\pmb{\beta_1}$}
                          (node312.center) -- (node313);
  \draw [grayfactor] (node313) -- (node314.center) -- (node315.center) -- (node316); 
  \draw [fullfactor, nicered] (node316) -- node [below left=1.4mm] {\footnotesize $\pmb{\beta_2}$} (node32) ;
  \draw [fullfactor, nicered] (node32) -- (node33.center) -- node [below left=0mm] {\footnotesize $\pmb{\beta_1}$}
                          (node34.center) -- (node35);
  \draw [fullfactor, nicered] (node35) -- (node36.center) -- node [below right=0.5mm] {\footnotesize $\pmb{\beta_3}$}
                          (node37.center) -- (node38);
  \draw [fullfactor, nicered] (node38) -- node [below left=1.4mm] {\footnotesize $\pmb{\beta_2}$} (node39) ;
  \draw [fullfactor, nicered] (node39) -- (node40.center) -- node [below left=0mm] {\footnotesize $\pmb{\beta_1}$}
                          (node41.center) -- (node42);
  \draw [fullfactor, nicered] (node42) -- (node43.center) -- node [below right=0.5mm] {\footnotesize $\pmb{\beta_3}$}
                          (node44.center) -- (node45);
  \draw [fullfactor, nicered] (node45) -- node [below left=1.4mm] {\footnotesize $\pmb{\beta_2}$} (node46) ;
  \draw [grayfactor] (node46) -- (node47.center) -- (node48.center) -- (node49); 
  \draw [fullfactor, nicered] (node49) -- (node491.center) -- node [below right=0.5mm] {\footnotesize $\pmb{\beta_3}$}
                          (node492.center) -- (node493);
  \draw [grayfactor] (node493) -- (node494.center) -- (node495.center) -- (node496); 
  
  \draw [fullfactor, gray] (node496) -- node [below=0.5mm] {\footnotesize $~\pmb{\gamma_1}$} (node50);
  \draw [fullfactor, gray] (node50) -- node [below=0.5mm] {\footnotesize $~\pmb{\gamma_1}$} (node51);
  \draw [fullfactor, gray] (node51) -- node [below=0.5mm] {\footnotesize $~\pmb{\gamma_1}$} (node52);
  \draw [grayfactor] (node52) -- (node53);
\end{scope}
\end{tikzpicture}

%% file: figures/inv.tex
\begin{tikzpicture}[baseline=0, inner sep=0, outer sep=0, minimum size=0pt, scale=0.5, yscale=0.9]
  \tikzstyle{dot} = [draw, circle, fill=white, minimum size=4pt]
  \tikzstyle{fulldot} = [draw, circle, fill=black!70, minimum size=4pt]
  \tikzstyle{grayfactor} = [->, shorten >=1pt, rounded corners=6, dashed]
  \tikzstyle{factor} = [->, shorten >=1pt, rounded corners=6]
  \tikzstyle{dotfactor} = [->, shorten >=1pt, dotted, rounded corners=6]
  \tikzstyle{fullfactor} = [->, >=stealth, shorten >=1pt, very thick, rounded corners=6]
  \tikzstyle{dotfullfactor} = [->, >=stealth, shorten >=1pt, dotted, very thick, rounded corners=6]

  \fill [pattern=north east lines, pattern color=gray!25]
        (4,-0.75) rectangle (8,9);
  \draw [dashed, thin, gray] (4,-0.75) -- (4,9);
  \draw [dashed, thin, gray] (8,-0.75) -- (8,9);
  \draw [gray] (4,-1) -- (4,-1.25) -- (8,-1.25) -- (8,-1);
  \draw [gray] (6,-1.5) node [below] {\footnotesize $I$};

  \fill [pattern=north east lines, pattern color=gray!25]
        (12,-0.75) rectangle (16,9);
  \draw [dashed, thin, gray] (12,-0.75) -- (12,9);
  \draw [dashed, thin, gray] (16,-0.75) -- (16,9);
  \draw [gray] (12,-1) -- (12,-1.25) -- (16,-1.25) -- (16,-1);
  \draw [gray] (14,-1.5) node [below] {\footnotesize $I'$};

  \draw (2,0) node (node0) {};
  \draw (4,0) node [dot] (node1) {};
  \draw (8,0) node [dot] (node2) {};
  \draw (12,0) node [dot] (node3) {};
  \draw (16,0) node [dot] (node4) {};
  \draw (18,0) node (node5) {};
  \draw (18,1) node (node6) {};

  \draw (16,1) node [dot] (node7) {};
  \draw (14,1) node (node8) {};
  \draw (14,2) node (node9) {};
  \draw (16,2) node (node10) {};
  \draw (17,2) node (node11) {};
  \draw (17,3) node (node12) {};
  \draw (16,3) node [fulldot] (node13) {};
  \draw (12,3) node [dot] (node14) {};

  \draw (8,3) node [dot] (node15) {};
  \draw (4,3) node [dot] (node16) {};
  \draw (2,3) node (node17) {};
  \draw (2,4) node (node18) {};

  \draw (4,4) node [dot] (node19) {};
  \draw (6,4) node (node20) {};
  \draw (6,5) node (node21) {};
  \draw (4,5) node [dot] (node22) {};
  \draw (3,5) node (node23) {};
  \draw (3,6) node (node24) {};
  \draw (4,6) node [fulldot] (node25) {};
  \draw (8,6) node [dot] (node26) {};

  \draw (12,6) node [dot] (node27) {};
  \draw (14,6) node (node28) {};
  \draw (14,7) node (node29) {};
  \draw (12,7) node [dot] (node30) {};

  \draw (8,7) node [dot] (node31) {};
  \draw (6,7) node (node32) {};
  \draw (6,8) node (node33) {};
  \draw (8,8) node [dot] (node34) {};

  \draw (12,8) node [dot] (node35) {};
  \draw (16,8) node [dot] (node36) {};
  \draw (18,8) node (node37) {};

  \draw [grayfactor] (node0) -- (node1);
  \draw [grayfactor] (node1) -- (node2);
  \draw [grayfactor] (node2) -- (node3);
  \draw [grayfactor] (node3) -- (node4);
  \draw [grayfactor] (node4) -- (node5.center) -- (node6.center) -- (node7);
  
  \draw [fullfactor,black!70] (node7) -- (node8.center) -- (node9.center) -- (node10);
  \draw [grayfactor] (node10) -- (node11.center) -- (node12.center) -- (node13);
  \draw [fullfactor,nicered] (node13) -- (node14);

  \draw [grayfactor] (node14) -- (node15);
  \draw [grayfactor] (node15) -- (node16);
  \draw [grayfactor] (node16) -- (node17.center) -- (node18.center) -- (node19);

  \draw [fullfactor,black!70] (node19) -- (node20.center) -- (node21.center) -- (node22);
  \draw [grayfactor] (node22) -- (node23.center) -- (node24.center) -- (node25);
  \draw [fullfactor,nicered] (node25) -- (node26);
  
  \draw [grayfactor] (node26) -- (node27);
  \draw [fullfactor,black!70] (node27) -- (node28.center) -- (node29.center) -- (node30);
  \draw [grayfactor] (node30) -- (node31);
  \draw [fullfactor,black!70] (node31) -- (node32.center) -- (node33.center) -- (node34);

  \draw [grayfactor] (node34) -- (node35);
  \draw [grayfactor] (node35) -- (node36);
  \draw [grayfactor] (node36) -- (node37);
  
  \draw (node13) node [above left=2mm] {\small {$\pmb{e'}\ \ \ $}};
  \draw (node25) node [above right=2mm] {\small $\ \ \ \pmb{e}$};
\end{tikzpicture}

%% file: figures/decomp-idemp.tex
\begin{tikzpicture}[baseline=0, inner sep=0, outer sep=0, minimum size=0pt]
  \tikzstyle{dot} = [draw, circle, fill=white, minimum size=4pt]
  \tikzstyle{fulldot} = [draw, circle, fill=black, minimum size=4pt]
  \tikzstyle{grayfactor} = [->, shorten >=1pt, rounded corners=5, dashed]
  \tikzstyle{factor} = [->, shorten >=1pt, rounded corners=5]
  \tikzstyle{dotfactor} = [->, shorten >=1pt, dotted, rounded corners=5]
  \tikzstyle{fullfactor} = [->, >=stealth, shorten >=1pt, very thick, rounded corners=5]
  \tikzstyle{dotfullfactor} = [->, >=stealth, shorten >=1pt, dotted, very thick, rounded corners=5]

\begin{scope}[xshift=5cm, yshift=-0.3cm, xscale=0.45, yscale=0.39]
  \tikzstyle{factor} = [->, shorten >=1pt, dotted, rounded corners=4]
  \tikzstyle{fullfactor} = [->, >=stealth, shorten >=1pt, very thick, rounded corners=4]

  \fill [pattern=north east lines, pattern color=gray!25]
        (4,-1) rectangle (8,17);
  \fill [pattern=north east lines, pattern color=gray!25]
        (8,1) rectangle (12,17);
  \fill [pattern=north east lines, pattern color=gray!25]
        (12,3) rectangle (16,17);
  \draw [dashed, gray] (4,-1) -- (4,17);
  \draw [dashed, gray] (8,-1) -- (8,17);
  \draw [dashed, gray] (12,-1) -- (12,17);
  \draw [dashed, gray] (16,-1) -- (16,17);
  \draw [gray] (4,-1.25) -- (4,-1.5) -- (7.9,-1.5) -- (7.9,-1.25);
  \draw [gray] (8.1,-1.25) -- (8.1,-1.5) -- (12,-1.5) -- (12,-1.25);
  \draw [gray] (12.1,-1.25) -- (12.1,-1.5) -- (15.9,-1.5) -- (15.9,-1.25);
  \draw [gray] (6,-1.75) node [below] {\footnotesize $I_1$};
  \draw [gray] (10,-1.75) node [below] {\footnotesize $I_2$};
  \draw [gray] (14,-1.75) node [below] {\footnotesize $I_3$};

 \draw (2,0) node (node0) {};
  \draw (4,0) node [dot, draw=black] (node1) {};

  \draw (6,0) node (node2) {};
  \draw (6,1) node (node3) {};
  \draw (4,1) node [dot, draw=black] (node4) {};
  \draw (2,1) node (node5) {};
  \draw (2,2) node (node6) {};

 \draw (1.5,2) node (node01) {};
  \draw (4,2) node [dot, draw=black] (node111) {};
  
 \draw (1.5,2) node (node01) {};
  \draw (4,2) node [dot, red] (node111) {};
  
  \draw (6,2) node (node211) {};
  \draw (6,3) node (node311) {};
  \draw (4,3) node [dot, draw=black] (node411) {};

  \draw (2,3) node (node511) {};
  \draw (2,4) node (node611) {};

  \draw (4,4) node [fulldot, black] (node7) {};
  \draw (8,4) node [dot, draw=black] (node8) {};

  \draw (10,4) node (node9) {};
  \draw (10,5) node (node10) {};
  \draw (8,5) node [dot, draw=black] (node11) {};
  \draw (6,5) node  (node12) {};
  \draw (6,6) node (node13) {};
  \draw (8,6) node [dot, red] (node1111) {};
  \draw (8,7) node [dot, draw=black] (node1411) {};
  \draw (6,7) node  (node1511) {};
  \draw (6,8) node (node1611) {};
   \draw (8,8) node [fulldot,blue] (node14) {};

  \draw (10,5) node (node911) {};
  \draw (10,6) node (node1011) {};
  \draw (10,6) node  (node1211) {};
  \draw (10,7) node (node1311) {};

  \draw (12,8) node [dot, blue] (node15) {};
  \draw (14,8) node  (node16) {};
  \draw (14,9) node  (node17) {};
  \draw (12,9) node [dot, draw=black] (node18) {};
  \draw (10,9) node (node19) {};
  \draw (10,10) node  (node20) {};
  \draw (12,10) node [dot, violet] (node2111) {};
  \draw (14,10) node (node1911) {};
  \draw (14,11) node  (node2011) {};

 \draw (12,11) node [dot, draw=black] (node3111) {};
 
 \draw (10,11) node  (node3112) {};
 
 \draw (10,12) node  (node3113) {};

  \draw (12,12) node [fulldot, black] (node21) {};
  \draw (16,12) node [dot, violet] (node22) {};
  \draw (18,12) node (node23) {};
  \draw (18,13) node  (node24) {};
  \draw (16,13) node [dot, draw=black] (node25) {};
  \draw (14,13) node  (node26) {};
  \draw (14,14) node (node27) {};
  \draw (16,14) node [dot, draw=black] (node28) {};

    \draw (18,14) node  (node2900) {};
  \draw (18,15) node (node2911) {};

  \draw (16,15) node [dot, draw=black] (node30) {};
  \draw (16,16) node [dot, draw=black] (node31) {};
  
  \draw (14,15) node (node3000) {};
  \draw (14,16) node  (node3100) {};
  
  \draw (18,16) node (node3133) {};
    \draw [grayfactor] (node31) -- (node3133);

  \draw [grayfactor] (node0) -- (node1);
  \draw [fullfactor] (node1) -- (node2.center) -- node [below right=0.5mm] {}
                           (node3.center) -- (node4); 
                           
  \draw [fullfactor, red] (node111) -- (node211.center) -- node [below right=0.5mm] {}
                           (node311.center) -- (node411);

  \draw [grayfactor] (node4) -- (node5.center) -- (node6.center) -- (node111);

  \draw [grayfactor] (node411) -- (node511.center) -- (node611.center) -- (node7);
  \draw [fullfactor] (node7) -- node [below=.7mm] {} (node8);

  \draw [fullfactor, black] (node8) -- (node9.center) -- node [below right=0.5mm] {}
                           (node10.center) -- (node11);
  \draw [fullfactor, red] (node11) -- (node12.center) -- node [below left=0.5mm] {}
                           (node13.center) -- (node1111);
                           
                             \draw [fullfactor] (node1111) -- (node1211.center) -- node [below right=0.5mm] {}
                           (node1311.center) -- (node1411);
                           
                           \draw [fullfactor, black] (node1411) -- (node1511.center) -- node [below right=0.5mm] {}
                           (node1611.center) -- (node14);

  \draw [fullfactor, blue] (node14) -- node [below=.7mm] {} (node15);

  \draw [fullfactor] (node15) -- (node16.center) -- node [below right=0.5mm] {}
                           (node17.center) -- (node18);
  \draw [fullfactor, black] (node18) -- (node19.center) -- node [below left=0.5mm] {}
                           (node20.center) -- (node2111);
                           
                            \draw [fullfactor, violet] (node2111) -- (node1911.center) -- node [below left=0.5mm] {}
                           (node2011.center) -- (node3111);

\draw [fullfactor, black] (node3111) -- (node3112.center) -- node [below left=0.5mm] {}
                           (node3113.center) -- (node21);
                           
                          \draw [fullfactor, black] (node30) -- (node3000.center) -- node [below left=0.5mm] {}
                           (node3100.center) -- (node31);

  \draw [fullfactor, violet] (node21) -- node [below=.7mm] {} (node22);
  \draw [grayfactor] (node22) -- (node23.center) -- (node24.center) -- (node25); 
  \draw [fullfactor, black] (node25) -- (node26.center) -- node [below left=0.5mm] {}
                           (node27.center) -- (node28);
\draw [grayfactor] (node28) -- (node2900.center) -- (node2911.center) -- (node30);

\end{scope}
\end{tikzpicture}

%% file: proof-bounded.tex
\section{Proof of Theorem~\ref{thm:oneway-def}}
\label{sec:main}

This section is devoted to proving the characterization of
one-way resynchronizability in the bounded-visit case. 
We will use the notion of \emph{bounded-traversal} 
from~\cite{KM20}, that was shown to characterize the class of
bounded regular resynchronizers, in as much as bounded-delay
characterizes rational resynchronizers~\cite{fjlw16icalp}.

\begin{definition}[traversal \cite{KM20}] 
Let $\s=(u,v,\orig)$ and $\s'=(u,v,\orig')$ be two synchronized pairs with 
the same input and output words. 

Given two input positions $y,y' \in\dom(u)$, 
we say that \emph{$y$ traverses $y'$} if there is a pair $(y,z)$
of source and target origins associated with the same output position
such that $y'$ is between $y$ and $z$,
with $y'\neq z$ and possibly $y'=y$.
More precisely:
\begin{itemize}
  \item $(y,y')$ is a \emph{left-to-right traversal} if
        $y\le y'$ and for some output position $x$, $\orig(x)=y$ and $z=\orig'(x)>y'$;
  \item $(y,y')$ is a \emph{right-to-left traversal} if
        $y\ge y'$ and for some output position $x$, $\orig(x)=y$ and $z=\orig'(x)<y'$.
\end{itemize}

A pair  $(\s,\s')$ of synchronized pairs with input $u$ and output $v$ 
is said to have \emph{$k$-bounded traversal}, with $k\in\Nat$,
if every $y' \in \dom(u)$ is traversed by at most $k$ distinct 
positions of $\dom(u)$. 

A resynchronizer $\Rr$ has \emph{bounded traversal} 
if there is some $k \in\Nat$ such that every 
$(\s,\s') \in \Rr$ has $k$-bounded traversal.
\end{definition}

\begin{lemma}[\cite{KM20}]\label{lem:bounded-traversal}
A regular resynchronizer is bounded if and only if it has bounded traversal. 
\end{lemma}

\begin{proof}[of Theorem \ref{thm:oneway-def}]
First of all, observe that the implication $4\rightarrow 1$ is straightforward.
To prove the implication $1\rightarrow2$, assume that there is a 
$k$-bounded, regular resynchronizer $\Rr$ that is $T$-preserving 
and such that $\Rr(T)$ is order-preserving.
Lemma \ref{lem:bounded-traversal} implies that $\Rr$ has $t$-bounded
traversal, for some constant $t$. 
We head towards proving that $T$ has cross-width bounded by $t+k$. 
Consider two synchronized pairs $\s=(u,v,\orig)$ 
and $\s'=(u,v,\orig')$ such that $\s\in \sem{T}_o$ and $(\s,\s') \in \Rr$, 
and consider a cross $(X_1,X_2)$ of $\s$. 
We claim that $|\orig(X_1)|$ or $|\orig(X_2)|$ is at most $t+k$.
Let $x_1=\min(\orig(X_1))$, $x'_1=\max(\orig'(X_1))$, 
$x_2=\max(\orig(X_1))$, and $x'_2=\min(\orig'(X_2))$. 
Since $(X_1,X_2)$ is a cross, we have $x_1>x_2$,
and since $\s'$ is order-preserving, we have $x'_1 \le x'_2$.
Now, if $x'_1>x_2$, then at least $|\orig(X_2)|-k$  input positions 
from $X_2$ traverse $x'_1$ to the right (the $-k$ term is due to the 
fact that at most $k$ input positions can be resynchronized to $x'_1$). 
Symmetrically, if $x'_1 \le x_2$, then at least $|\orig(X_1)|-k$ input 
positions from $X_1$ traverse $x_2$ to the left (the $-k$ term accounts 
for the case where some positions are resynchronized to $x'_1$ and $x'_1=x_2$). 
This implies $\min(|\orig(X_1)|,|\orig(X_2)|) \le t+k$, as claimed.

The remaining implications rely on the assumption that $T$ 
is bounded-visit.

The implication $2\rightarrow3$ is shown by contraposition: 
one considers a successful run $\r$ with an inversion, and shows
that crosses of arbitrary width emerge after pumping the loops of 
the inversion (here Lemma~\ref{lem:order-preservation} is crucial).

The proof of $3 \rightarrow 4$ is more involved,
we only sketch it here. Assuming that no successful run of $T$ 
has inversions we build a  partially bijective, 
regular resynchronizer $\Rr$ that is $T$-preserving and 
$\Rr(T)$ is order-preserving. 
The resynchronizer $\Rr$ uses some parameters
to guess a successful run $\r$ of $T$ on $u$ and a 
factorization tree of bounded height for $\r$.
Formally, a \emph{factorization tree} for a sequence $\alpha$
of monoid elements (e.g.~the flows $\flow_\r([y,y])$ for all input positions $y$)
is an ordered, unranked tree whose yield is the sequence $\alpha$. \label{p:factorization}
The leaves of the factorization tree are labeled with the elements of $\alpha$. 
All other nodes have at least two children and
are labeled by the monoid product of the child labels (in our case
by the flows of $\r$ induced by the covered factors in the input).
In addition, if a node has more than two children, then all its children 
must have the same label, representing an idempotent element of the monoid. 
By Simon's factorization theorem~\cite{sim72}, every sequence of monoid
elements has some factorization tree of height at most linear in the size 
of the monoid (in our case, at most $3|M_T|$, see e.g.~\cite{col07}). 

\emph{Parameters.} We use input parameters to encode the successful 
run $\r$ and a factorization tree for $\r$ of height at most $H=3|M_T|$.
These parameters specify, for each input interval corresponding to a
subtree, the start and end positions of the interval and the label of 
the root of the subtree.
Correctness of these annotations can be enforced by an MSO sentence 
$\ipar$.
%
The run and the factorization tree also need to be encoded over the
output, using output parameters. More precisely, given a level in the
tree and an output position, we need to be able to determine the flow
and the productive edge that generated that position. 
%
\longshort{
The technical details for checking correctness of the
output annotation using the formulas $\opar$, $\move_\otype$ and
$\nxt_{\otype,\otype'}$ can be found in Appendix~\ref{ap:main}. 
}{
We omit the technical details for checking correctness of the
output annotation using the formulas $\opar$, $\move_\otype$ and
$\nxt_{\otype,\otype'}$.
}

\emph{Moving origins.}
For each level $\ell$ of the factorization tree, a partial resynchronization
relation $\Rr_\ell$ is defined. The relation is partial in the sense that 
some output positions may not have a source-target origin pair defined
at a given level. But once a source-target pair is defined for some
output position at a given level, it remains defined for all higher levels.

In the following we write $\bout{\nd}$ for the dominant output interval
associated with the input interval $I(\nd)$ corresponding to a node $\nd$ 
in the tree.  For every level $\ell$ of the factorization tree, the
resynchronizer $\Rr_\ell$ will be a partial function from source
origins to target origins, and will satisfy the following:
\begin{itemize}
  \item the set of output positions for which $\Rr_\ell$ defines target origins 
        is the union of the intervals $\bout{\nd}$ for all
        nodes $\nd$ at level $\ell$;
  \item $\Rr_\ell$ only moves origins within the same interval at level $\ell$,
        that is, $\Rr_\ell$ defines only pairs $(y,z)$ of source-target origins 
        such that $y,z\in I(\nd)$ for some node $\nd$ at level $\ell$;
  \item the target origins defined by $\Rr_\ell$ are order-preserving
        within every interval at level $\ell$, that is, 
        for all output positions $x<x'$, if $\Rr_\ell$ defines 
        the target origins of $x,x'$ to be $z,z'$, 
        respectively, and if $z,z'\in I(\nd)$ for some node $\nd$ at level $\ell$,
	    then $z\le z'$;
  \item $\Rr_\ell$ is $\ell\cdot4K\obound$-bounded, namely, there are at most
        $\ell\cdot4K\obound$ distinct source origins that are moved by $\Rr_\ell$
        to the same target origin.
\end{itemize}
The construction of $\Rr_\ell$ is by induction on $\ell$. 
For a binary node $\nd$ at level $\ell$ with children
$\nd_1,\nd_2$, the resynchronizer $\Rr_\ell$ inherits the 
source-origin pairs from level $\ell-1$ for output positions
that belong to $\bout{\nd_1}\cup\bout{\nd_2}$. 
Note that $\bout{\nd_1} < \bout{\nd_2}$ by Lemma~\ref{lem:order}, 
so $\Rr_\ell$ is order-preserving inside $\bout{\nd_1} \cup \bout{\nd_2}$. 
Output positions inside 
$\bout{\nd} \setminus (\bout{\nd_1} \cup \bout{\nd_2})$ 
are moved in an order-preserving manner to one of the
extremities of $I(\nd)$, or to the last position of
$I(\nd_1)$. Boundedness of $\Rr_\ell$ is guaranteed 
by Lemma~\ref{lem:bin}.

%% file: conclusions.tex
\section{Complexity}\label{sec:complexity}

We discuss the effectiveness and complexity of our characterization.
For a  $k$-visit transducer $T$, 
the effectiveness of the characterization relies on detecting inversions 
in successful runs of $T$. It is not difficult to see that this can be
decided in space that is polynomial in the size of $T$ and the
bound $k$.  We can also show that
one-way resynchronizability is \PSPACE-hard. For this we recall that the emptiness problem for two-way finite
automata is \PSPACE-complete. Let $A$ be a two-way automaton accepting
some language $L$, and let $\Sigma$ be a binary alphabet disjoint from that of $L$.
The function $\{(w \cdot a_1\dots a_n, a_n\dots a_1) \mid w\in L,
a_1\dots a_n \in \Sigma^*, n \ge 0\}$
can be realized by a two-way transducer $T$ of size polynomial in $|A|$,
and $T$ is one-way resynchronizable if and only if $L$ is empty.

In the unrestricted case, we showed that one-way resynchronizability
is decidable (Theorem \ref{thm:decidability}).  We briefly outline the
complexity of the decision procedure:
\begin{enumerate}
  \item First one checks
        that $T$ is $K$-sparse for some $K$. 
		To do this, we construct from $T$ the regular language $L$ 
		of all inputs with some positions marked that correspond to 
		origins produced within the same vertical loop. 
		Bounded sparsity is equivalent to having a uniform bound 
		on the number of marked positions in every input from $L$.
        Standard techniques for two-way automata 
        allow to decide this in space that is polynomial
        in the size of $T$. Moreover, this also gives us a
        computable exponential bound to the largest constant 
        $K$ for which $T$ can be $K$-sparse. 
  \item Next, we construct from the $K$-sparse transducer $T$  
        a bounded-visit transducer $T'$ that 
        is classically equivalent to $T$ and has exponential size.
  \item Finally, we decide one-way resynchronizability 
        of $T'$ by detecting  inversions in successful runs of $T'$ (Theorem~\ref{thm:oneway-def}).
\end{enumerate}
Summing up, one can decide one-way resynchronizability of unrestricted 
two-way transducers in exponential space.
It is  open if this bound is optimal.  
We also do not have any interesting bound on the size of the
resynchronizer that witnesses one-way resynchronizability,
both in the bounded-visit case and in the unrestricted case.
Similarly, we lack upper and lower bounds on the size of the 
resynchronized one-way transducers, when these exist.

\section{Conclusions}\label{sec:conclusions}
As the main contribution of this paper,        
we provided a characterization for the subclass of 
two-way transducers that are one-way resynchronizable,
namely, that can be transformed by some bounded, regular 
resynchronizer, into an origin-equivalent one-way transducer.

There are similar definability problems that 
emerge in the origin semantics. For instance, one could ask whether
a given two-way transducer can be resynchronized, through 
some bounded, regular resynchronization, to a relation that 
is origin-equivalent to a first-order transduction.
This can be seen as a relaxation of the first-order 
definability problem in the origin semantics, namely,
the problem of telling whether a two-way transducer is
origin-equivalent to some first-order transduction, shown decidable in \cite{boj14icalp}. 
It is worth contrasting the latter problem with the challenging 
open problem whether a given transduction is equivalent to a first-order 
transduction in the classical setting.

\medskip

\paragraph{Acknowledgments.} 
We thank the FoSSaCS reviewers for their constructive
and useful comments.


%% file: appendix-resynchs.tex
\section{Proofs from Section~\ref{ssec:regresync}}\label{app:resynchs}

\OneBounded*

\begin{proof}
    	Let $R=(\bar I,\bar O,\ipar,\opar,(\move_\otype)_\otype,(\nxt_{\otype,\otype'})_{\otype,\otype'})$
	be a $k$-bounded, regular resynchronizer. Let $\hat{u}$ and
        $\hat{v}$ be a pair of annotated 
	input and output satisfying $\ipar$ and $\opar$ respectively.
	To construct an equivalent $1$-bounded regular resynchronizer $R'$ 
	we introduce additional output parameters. 
	Specifically, each output position
	will be annotated with an output type $\t$ from $R$ and an additional
	index in $\{1,\dots,k\}$. 
	The intended meaning of the index is as follows:
	if $(y,z)$ is the source/target origin pair associated with
        an output position labeled by $(\t,i)$, $i \in \set{1,\dots,k}$, then
	then there are exactly $(i-1)$ positions $y'<y$ such that 
	$(\hat{u},y',z)\models \move_\t$.
	
	Note that this indexing depends on the choice of the target origin $z$. 
	Therefore, different indexing are possible for different choice of 
	the target origin $z$.
      
	
	Based on the resynchronizer $R$, we define the new resynchronizer as
	$R'=(\bar I,\bar O',\ipar,\opar',(\move'_{(\otype,i)})_{\otype,i},
	(\nxt_{(\otype,i),(\otype',i')})_{\otype,i,\otype',i'})$,
	where
	\begin{itemize}
		\item $\bar O'=\bar O \uplus \set{O'_1,\dots,O'_k}$ consists of the old output 
		parameters $\bar O$ of $R$ plus some new parameters $O'_1,\dots,O'_k$ for
		representing indices in $\{1,\dots,k\}$;
		\item $\opar'$ defines language of all output annotations whose projections
		over $\G'$ (the output alphabet extended with the parameters of $R$) 
		satisfy $\opar$ and each position is marked by exactly one index;
		\item given a type $\otype'$ that encodes a type $\otype$ of $R$
		and an index $i\in\{1,\dots,k\}$, $\move'_{\otype'}(y,z)$
		states that $y$ is the $i$-th position $y'$ satisfying $\move_\otype(y',z)$;
		This property can be expressed by the MSO-formula
		\begin{align*}
		\exists ~y_1<\dots<y_i=y ~ \bigwedge\nolimits_j \move_\otype(y_j,z) \\
		~\wedge~ \forall y'\leq y ~ \big( \move_\otype(y',z) \rightarrow
		\bigvee\nolimits_j y'=y_j \big);
		\end{align*}
		\item $\nxt'_{(\otype,i),(\otype',i')}(z,z')$ enforces the same property 
		as $\nxt_{\otype,\otype'}(z,z')$.
	\end{itemize}
	
	The resynchronizer $R'$ is $1$-bounded by definition of $\move'_{(\otype,i)}$. 
	If for positions $y<y'$, $(\hat{u},y,z)\models \move'_{(\otype,i)}$ and 
	$(\hat{u},y',z)\models\move'_{(\otype,i)}$, then $y$ and $y'$ are
        both the $i$-th source position
	in $\hat{u}$ satisfying $\move_{\otype}$ with target $z$, which is a contradiction. 
	
	We now prove that $R$ and $R'$ define the same relation between synchronized pairs.
	First we show $R'\subseteq R$. Consider $((u,v),(u,v'))\in R'$. Therefore, there exists
	$\hat{u}\models \ipar$ and $\hat{v}\models\opar'$ such that $\move'$ applied to positions of 
	$\hat{v}$ give the $v'$ witnessing $((u,v),(u,v'))\in R'$.
	By definition of $\opar'$, $\hat{v}_{\G'}\models \opar$. 
	Suppose, a position $x$ of output type $(\t,i)$ is moved from origin $y$ in $v$ to 
	$z$ in $v'$. 
	This means $(\hat{u},y,z)\models \move'_{(\t,i)}$. Then, by definition of
	$\move'_{(\t,i)}$, $(\hat{u},y,z)\models \move_\t$.
	This shows $R'\subseteq R$.
	
	For the  containment $R\subseteq R'$, consider $((u,v),(u,v'))\in R$.
	Therefore, there exists $\hat{u}\models\ipar$ and $\hat{v}\models\opar$ such that $\move$
	applied to each position in $\hat{v}$  witnesses $((u,v),(u,v'))\in R$.
	This means for every position $x\in \dom(\hat{v})$ with output-type $\t$, there exist
	$y,z$, such that $(\hat{u},y,z)\models \move_{\t}$, 
	$y=\orig(v(x))$ and $z=\orig(v'(x))$. For such a position $x\in \dom(\hat{v})$ of output
	type $\t$, let $i\in \set{1,\dots,k}$ be such that there are exactly $i-1$ positions
	$y_1<y_2<\dots y_{i-1}<y$ such that $(\hat{u},y_j,z)\models\move_\t$.
	Let $\hat{v}'$ be the annotation of $\hat{v}$ where every
        position $x$ is annotated with 
	the index $i$ as above. Clearly $\hat{v}'\models \opar'$ and therefore, 
	$((u,v),(u,v'))\in R'$. We conclude $R=R'$.
\qed
\end{proof}

\ResyncClosure*

\begin{proof}
  	Let  
	$R=(\bar I,\bar O,\ipar,\opar,(\move_\otype)_\otype,(\nxt_{\otype,\otype'})_{\otype,\otype'})$
	and 
	$R'=(\bar I',\bar O',\ipar',\opar'$, $(\move'_\lambda)_\lambda,
	(\nxt'_{\lambda,\lambda'})_{\lambda,\lambda'})$ 
	be two bounded, regular resynchronizers. 
	In view of Lemma \ref{lem:one-bounded}, we can assume that 
	both resynchronizers are $1$-bounded. 
	The composition $R\circ R'$ can be defined by combining the effects of $R$ 
	and $R'$ almost component-wise.
	Some care should be taken, however, in combining the formulas $\nxt$ and $\nxt'$.
	Formally, we define the composed resynchronizer
	$R'' = (\bar I'',\bar O'',\ipar'',\opar'',
	(\move''_{(\otype,\lambda)})_{\otype,\lambda},
	(\nxt''_{(\otype,\lambda),(\otype',\lambda')})_{\otype,\lambda,\otype',\lambda'})$,
	where
	\begin{itemize}
		\item $\bar I''$ is the union of the parameters $\bar I$ and $\bar I'$,
		\item $\bar O''$ is the union of the parameters $\bar O$ and $\bar O'$,
		\item $\ipar''$ is the conjunction of the formulas $\ipar$ and $\ipar'$;
		\item $\opar''$ is the conjunction of the formulas $\opar$ and $\opar'$;
		\item $\move''_{(\otype,\lambda)}(y,z)$ states the existence of some position
		$t$ satisfying both formulas $\move_\otype(t,z)$ and
		$\move'_\lambda(y,t)$;
		\item $\nxt_{(\otype,\lambda),(\otype',\lambda')}(z,z')$ requires that
		$\nxt_{\otype,\otype'}(z,z')$ holds and, moreover, that there exist
		some positions $t,t'$ satisfying $\move_\otype(t,z)$,
		$\move_{\otype'}(t',z')$, and $\nxt_{\lambda,\lambda'}(t,t') $;
		note that these positions $t,t'$ are uniquely determined from $z,z'$
		since $R$ is $1$-bounded, and they act, at the same time,
		as source origins for $R$ and as target origins for $R'$.
	\end{itemize}
	
	By definition, $\move''_{(\t,\lambda)}$ is $1$-bounded, thus $z$ and $\t$ 
	determine a unique $t$, which together with $\lambda$ determines a unique $y$.
	It is also easy to see that $R''$ is equivalent to $R\circ R'$ as the positions
	corresponding to $t$ in formulas $\move''_{(\t,\lambda)}$ 
	and $\nxt''_{(\t,\lambda),(\t',\lambda')}$
	correspond to the source origin of $R$ and target origin of $R'$.
\qed
\end{proof}


%% file: appendix-flows.tex
\section{Proofs from Section~\ref{ssec:flows}} 
\label{app:lem-nasty}

\OrderPreservation*

\begin{proof}
It is convenient to rephrase the claim of the lemma 
in terms of a juxtaposition operation on flows and 
in terms of an induced accessibility relation on edges. 
Formally, given two flows $F,G$, we define the juxtaposition $F G$ 
in a way similar to concatenation, with the only exception that 
in the result we maintain as an additional group of vertices the $\rgt$ vertices of $F$, 
glued with the state-matching $\lft$ vertices of $G$ 
(strictly speaking, the result of a juxtaposition of two flows is
not a flow, since it has three distinguished groups of vertices). 
We denote by $E\dots E$ the $n$-fold juxtaposition of the flow 
$E$ with itself (this must not be confused with the $n$-fold 
concatenation $E\cdot\ldots\cdot E$).

Let $F,E,G$, $f,e$, $\r_f,\r_e$ be as stated in the lemma, and let 
$\preceq$ denote the accessibility order between edges 
in a juxtaposition of flows, e.g.~$F E G$ (note that,
due to the type of flows considered here, $\preceq$ 
turns out to always be a total order on $F E G$).
Observe that the relative occurrence order of $\r_f$ and $\r_e$ 
inside $\r$ is the same as the accessibility order $\preceq$ of 
the edges $f$ and $e$ on the graph $F E G$. 
A similar claim holds for the occurrence order of $\r_f$ 
and any copy of $\r_e$ inside the pumped run $\pump^n_I(\r)$, which corresponds 
to the accessibility order of $f$ and any copy of $e$ in the graph $F E\dots E G$. 
Thanks to these correspondences, to prove the lemma it suffices 
to consider any copy $e'$ of $e$ in $F E\dots E G$, and show that 
\[
  f \preceq e \text{ in } F E G
  \qquad\text{iff}\qquad
  f \preceq e' \text{ in } F E\dots E G.
\]
We thus prove the above claim.
Consider the maximal path $\pi$ inside $E\dots E$ that contains the edge $e'$.
Note that this path starts and ends at some extremal vertices of $E\dots E$
(otherwise the path could be extended while remaining inside $E\dots E$).
Also recall that concatenation can be defined from juxtaposition by removing 
the intermediate groups of vertices, leaving only the extremal ones, and
by shortcutting paths into edges. We call this operation \emph{flattening}, 
for short.
In particular, since $E$ is idempotent, we have that $E = E\cdot\ldots\cdot E$
can be obtained from the flattening of $E\dots E$, and this operation transforms 
the path $\pi$ into an edge $e''$.
By construction, we have that 
$f \preceq e'$ in $F E\dots E G$
if and only if
$f \preceq e''$ in $F E G$.
So it remains to prove that 
\[
  f \preceq e \text{ in } F E G
  \qquad\text{iff}\qquad
  f \preceq e'' \text{ in } F E G.
\]
Clearly, this latter claim holds if
the edges $e$ and $e''$ coincide. This is indeed the case when $e$ is a straight
edge and $E$ is idempotent. The formal proof that this holds is rather tedious, 
but follows quite easily from a series of results we have already proven
in \cite{bgmp18}. Roughly speaking, one proves that:
\begin{itemize}
  \item the edges of an idempotent flow $E$ can be grouped into \emph{components} 
        (cf.~Definition 6.4 from \cite{bgmp18}),
        so that each component contains exactly one straight edge
        (cf.~Lemma 6.6 from \cite{bgmp18}, see also 
        the figure at page \pageref{fig:pumping},
         where components are represented by colors);
  \item every path inside the juxtaposition $E E$, with $E$ idempotent,
        consists of edges from the same component, say $C$;
        moreover, after the flattening from $E E$ to $E$, this path
        becomes an edge of $E$ that belongs again to the 
        component $C$
        (cf. Claims 7.3 and 7.4 in the proof of 
         Theorem 7.2 from \cite{bgmp18});
  \item every maximal path in $E\dots E$ that contains a straight edge
        starts and end at opposite sides of $E\dots E$
        (simple observation based on the definition of concatenation
         and Lemma 6.6 from \cite{bgmp18}).
\end{itemize}
To conclude, recall that $\pi$ is a path inside $E\dots E$ that
contains a copy $e'$ of the straight edge $e$, and that becomes 
the edge $e''$ after the flattening into $E$.
The previous properties immediately imply that $e=e''$.
\end{proof}


%% file: appendix-big.tex
\section{Proofs from Section~\ref{ssec:big}} \label{appendix:B}

As explained in the body, the technical lemmas involving 
output blocks are only applicable to transducers that avoid inversions.
\emph{Hereafter we assume that $T$ is a transducer
that avoids inversions, and we denote by $\r$ an arbitrary 
successful run of $T$.}

\Order*

\begin{proof}
$B_1$ and $B_2$ are clearly disjoint.
By way of contradiction, assume that $B_1$ and $B_2$ are $\obound$-large, 
but $B_1>B_2$. 
By Lemma~\ref{lem:bound}, we can find for both $i=1$ and $i=2$
a loop $J_i \subseteq I_i$ 
and a productive straight edge
$e_i\in\flw{J_i}$ 
that is witnessed by a subrun
intersecting $B_i$. 
Clearly, we have $J_1<J_2$, and since $B_1>B_2$, the subrun witnessing $e_1$ 
follows the subrun witnessing $e_2$.
Thus, $(J_1,e_1,J_2,e_2)$ is an inversion of $\r$, which contradicts 
the assumption that $T$ avoids inversions.
\end{proof}

\medskip
We now turn to the proofs of Lemmas \ref{lem:bin} and \ref{lem:idem},
which both require auxiliary lemmas relying again on the assumption
that $T$ avoids inversion.

\begin{restatable}{lemma}{Between}\label{lem:between}
Let $I$ be an input interval, $B_1<B_2$ two output blocks of $I$,
and $S$ the set of output positions strictly between $B_1$ and $B_2$
and with origins outside $I$.
If $B_1,B_2$ are $\obound$-large, then $S$ is $2\obound$-small.
\end{restatable}

\begin{proof}
By way of contradiction, suppose that  $S$ 
is  $2\obound$-large.
This means that 
$|\orig(S) \cap I'| > \obound$ for some interval 
$I'$ disjoint from $I$, say $I'<I$ 
(the case of $I'>I$ is treated similarly).
By Lemma \ref{lem:bound}, we can find two loops 
$J\subseteq I$ and $J'\subseteq I'$ 
and some productive straight edges 
$e\in\flw{J}$ and $e'\in\flw{J'}$ that are witnessed 
by subruns intersecting $B_1$ and $S$, respectively.
Since $S>B_1$, we know that the subrun witnessing
$e$ follows the subrun witnessing $e'$.
As in the previous proof, this shows the inversion
$(J,e,J',e')$, which contradicts 
the assumption that $T$ avoids inversions.
\end{proof}

\Bin*
      
\begin{proof}
By Lemma \ref{lem:order}, we have $B_1 < B_2$.
Moreover, all $\obound$-large output blocks of $I_1$ or $I_2$ 
are also $\obound$-large output blocks of $I$, so $B$ contains both $B_1$ and $B_2$.
Let $I_0$ be the maximal interval to the left of $I_1$, and thus adjacent to it,
and, similarly, let $I_3$ be the maximal interval to the right of $I_2$, and 
thus adjacent to it.

Suppose, by way of contradiction, that $B\setminus(B_1\cup B_2)$ is $4K\obound$-large.
This means that there is a $2K\obound$-large set $S\subseteq B\setminus(B_1\cup B_2)$ 
with origins entirely inside $I_0\juxt I_1$ or entirely inside $I_2\juxt I_3$. 
Suppose, w.l.o.g., that the former case holds, and decompose $S$ as a union of 
maximal output blocks $B'_1,B'_2,\dots,B'_n$ of either $I_0$ or $I_1$.
Since $S\cap B_1 = \emptyset$, we have that every block $B'_i$ 
with origins inside $I_1$ is $\obound$-small.
Similarly, by Lemma \ref{lem:between}, every block $B'_i$ 
with origins inside $I_0$ is $\obound$-small too.
Moreover, since $\r$ is $K$-visiting, we have that the number
$n$ of maximal output blocks of either $I_0$ or $I_1$ that 
are contained in $S$ is at most $2K$.
All together, this contradicts the assumption that $S$ is $2K\obound$-large.
\end{proof}

\medskip
  
\begin{restatable}{lemma}{Straight}\label{lem:straight}
Let $I$ be a loop of $\r$. Then $\flw{I}$ has at most one 
productive straight edge, and this edge must be $\lft\rgt$. 
\end{restatable}

%

\begin{proof}
 Suppose, by way of contradiction, that there are two productive
 straight edges in $\flw{I}$, say $e$ and $f$, with $e$ before $f$ in
 $\r$ (the reader may refer again to 
 the figure at page \pageref{fig:pumping},
 and think of $e$ and $f$, for instance, as the edges labeled by 
 $\a_2$ and $\g_1$, respectively). 
 Suppose that we pump $I$ twice, and let
 $I_1<I_2$ be the copies of $I$ in the pumped run $\r'$. Let also
 $e_1,e_2$ (resp.~$f_1,f_2$) be the corresponding copies of $e$ 
 (resp.~$f$), so that $e_j,f_j$ belong to the flow $\flw[\r']{I_j}$. It is
 easy to check the following properties:
 \begin{itemize}
 \item if $e$ is an $\lft\rgt$ edge, then the subrun witnessed by
   $e_1$ occurs in $\r'$ before the subrun witnessed by $e_2$ 
   (and the other way around if $e$ is $\rgt\lft$);
   \item the subruns witnessed by $e_1$ and $e_2$ occur in $\r'$ before the
     subruns witnessed by $f_1,f_2$ (this property follows easily from
     the observation that when building the product $\flw{I} \cdot
     \flw{I}$, the edges $e_1,e_2$ will be ``part'' of the edge $e$ in
     the product, whereas $f_1,f_2$ will be ``part'' of the edge
     $f$).
   \end{itemize}
   Let us assume first that $e$ is an $\rgt\lft$ edge. Then observe
   that $(I_1,e_1,I_2,e_2)$ is an inversion in $\r'$. But this
   contradicts $T$ being inversion-free. Therefore, both $e,f$ are
   $\lft\rgt$ edges. But then, $(I_1,f_1,I_2,e_2)$ is an inversion in
   $\r'$, and we have again a contradiction.
\end{proof}

\begin{remark}\label{rem:pump}
  The statement of Lemma~\ref{lem:straight} can be strengthened by
  observing the following property of productive edges in an
  idempotent flow. Assume that $I$ is a loop and $e$ is
  the unique productive straight edge in $\flw{I}$. 
  Let $f$ be some productive (non-straight) edge 
   of $\flw{I}$ with $f \not= e$. When $I$ is pumped then the subruns
   witnessing the copies of $f$ are part of the subrun witnessing
   $e$ in the product flow. This means for example, that in 
   the figure on page~\pageref{l:pumping} the
   productive edges are either all among the blue edges, or all among
   the gray edges (none of the red edges can be productive, because the
   straight edge is $\rgt\lft$, and would result in a productive $\rgt \lft$ edge 
   on pumping).
\end{remark}

\Idem*

\begin{proof}
By Lemma~\ref{lem:straight}, we can assume that $\flw{I}=\flw{I_k}$ has a unique
productive straight edge $e$, which is an $\lft\rgt$ edge.   Let $B'_k$ be the
output block corresponding to $e$ in $\flw{I_k}$. Since $\flw{I}$
is idempotent, 
  any output block of $I$ has one of the following shapes (see also
  Remark~\ref{rem:pump}):
\begin{itemize}
\item[(a)] A block $B =B'_1 \juxt J'_1 \juxt \dots J'_{n-1} \juxt B'_n$, for 
  some intervals $J'_1,\dots,J'_{n-1}$ such that
  $\out{I_k}$ is included in $J'_{k-1}\juxt B'_k\juxt J'_k$ for all $1< k <n$,
\item[(b)] At most $2K$ output blocks $L_1,\dots,L_p,R_1,\dots,R_s$, where
  each $L_i$ and $R_j$ corresponds to an edge of $\flw{I_1}$ and
  $\flw{I_n}$, respectively: the blocks $L_i, R_j$ appear before,
  respectively after the straight edge. 
\end{itemize}
Moreover, the order of the output blocks of $I$ is
$L_1,\dots,L_p,B,R_1,\dots,R_s$.  To illustrate the statement (a) above,
the reader can take as example $p=s=2$,
  $L_1=\a$, $L_2=\beta$, $R_1=\kappa$, $R_2=\zeta$, 
    $J'_1=\cdots = J'_{n-1}=\a\kappa \beta \zeta$ in Figure~\ref{fig:idempotent-loop-appendix}. 
 For statement (b), notice that in $I_1 \cdot I_2 \cdot I_3$, we have 
 the output blocks $L_1=\alpha, L_2=\beta$ of $I_1$, the straight edge 
    $(\gamma \alpha \kappa \beta \zeta)^2 \gamma$ (the purple zigzag) followed by $R_1=\kappa, R_2=\zeta$ 
     of $I_3$.  
    
Note that $B_k=\bout{I_k}$ is contained in $J'_{k-1}\juxt B'_k \juxt J'_k$ 
for all $1< k <n$. Moreover, $B_1=\bout{I_1}$ is contained in
$L_1  \cdots  L_p \juxt B'_1\juxt  J'_1$, and $B_n=\bout{I_n}$ 
is contained in
$J'_{n-1} \juxt B'_n \juxt R_1 \cdots  R_s$. 
Also by Lemma \ref{lem:order},
$B_j$ precedes $B_{j+1}$ for all $j$. 
\input{figures/decomp-idemp-appendix}

If one of the $L_k$ is $\obound$-large, then $B_1$ is non-empty, hence
$\bout{I}$ is non-empty and starts at the first position of
$B_1$. Similarly, if one of the $R_k$ is $\obound$-large then $B_n$ is
non-empty, hence $\bout{I}$ is non-empty and ends with the last
position of $B_n$. Otherwise, if all $L_j,R_j$ are $\obound$-small
then $\bout{I}$ is either empty or equal to $B$. In all cases we can write
$\bout{I}=B_1 \juxt J_1 \juxt B_2 \juxt J_2 \juxt \dots \juxt J_{n-1}
\juxt B_n$, with each $J_k$ consisting of at most $K$ $\obound$-small
blocks of $I_k$ and $K$ $\obound$-small blocks of
$I_{k+1}$, namely those left over after gathering the $\obound$-large
blocks into $\bout{I_k}$ and $\bout{I_{k+1}}$, respectively. Therefore, each $J_k$ is $2K\obound$-small.
\end{proof}

%% file: figures/decomp-idemp-appendix.tex
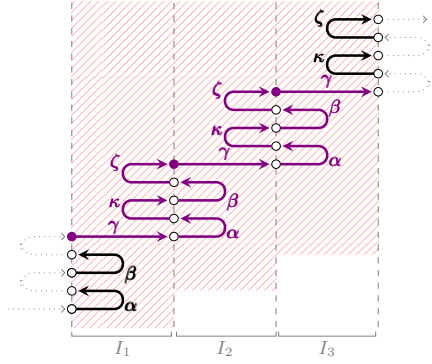
\begin{figure}[!t]
\centering
\scalebox{0.8}{%
\begin{tikzpicture}[baseline=0, inner sep=0, outer sep=0, minimum size=0pt]
  \tikzstyle{dot} = [draw, circle, fill=white, minimum size=4pt]
  \tikzstyle{fulldot} = [draw, circle, fill=black, minimum size=4pt]
  \tikzstyle{grayfactor} = [->, shorten >=1pt, rounded corners=6, gray, thin, dotted]
  \tikzstyle{factor} = [->, shorten >=1pt, rounded corners=6]
  \tikzstyle{dotfactor} = [->, shorten >=1pt, dotted, rounded corners=6]
  \tikzstyle{fullfactor} = [->, >=stealth, shorten >=1pt, very thick, rounded corners=6]
  \tikzstyle{dotfullfactor} = [->, >=stealth, shorten >=1pt, dotted, very thick, rounded corners=6]

\begin{scope}[xshift=5cm, yshift=-0.3cm, xscale=0.42, yscale=0.3]
  \tikzstyle{factor} = [->, shorten >=1pt, dotted, rounded corners=4]
  \tikzstyle{fullfactor} = [->, >=stealth, shorten >=1pt, very thick, rounded corners=4]

  \fill [pattern=north east lines, pattern color=red!25]
        (4,-1) rectangle (8,17);
  \fill [pattern=north east lines, pattern color=red!25]
        (8,1) rectangle (12,17);
  \fill [pattern=north east lines, pattern color=red!25]
        (12,3) rectangle (16,17);
  \draw [dashed, thin, gray] (4,-1) -- (4,17);
  \draw [dashed, thin, gray] (8,-1) -- (8,17);
  \draw [dashed, thin, gray] (12,-1) -- (12,17);
  \draw [dashed, thin, gray] (16,-1) -- (16,17);
  \draw [gray] (4,-1.25) -- (4,-1.5) -- (7.9,-1.5) -- (7.9,-1.25);
  \draw [gray] (8.1,-1.25) -- (8.1,-1.5) -- (12,-1.5) -- (12,-1.25);
  \draw [gray] (12.1,-1.25) -- (12.1,-1.5) -- (15.9,-1.5) -- (15.9,-1.25);
  \draw [gray] (6,-1.75) node [below] {\footnotesize $I_1$};
  \draw [gray] (10,-1.75) node [below] {\footnotesize $I_2$};
  \draw [gray] (14,-1.75) node [below] {\footnotesize $I_3$};

 \draw (1.5,0) node (node0) {};
  \draw (4,0) node [dot, draw=black] (node1) {};
  
 \draw (1.5,0) node (node0) {};
  \draw (4,0) node [dot, draw=black] (node1) {};

  \draw (6,0) node (node2) {};
  \draw (6,1) node (node3) {};
  \draw (4,1) node [dot, draw=black] (node4) {};
  \draw (2,1) node (node5) {};
  \draw (2,2) node (node6) {};

 \draw (1.5,2) node (node01) {};
  \draw (4,2) node [dot, draw=black] (node111) {};
  
 \draw (1.5,2) node (node01) {};
  \draw (4,2) node [dot] (node111) {};

  \draw (6,2) node (node211) {};
  \draw (6,3) node (node311) {};
  \draw (4,3) node [dot, draw=black] (node411) {};
  \draw (2,3) node (node511) {};
  \draw (2,4) node (node611) {};

  \draw (4,4) node [fulldot, violet] (node7) {};
  \draw (8,4) node [dot, draw=black] (node8) {};

  \draw (10,4) node (node9) {};
  \draw (10,5) node (node10) {};
  \draw (8,5) node [dot, draw=black] (node11) {};
  \draw (6,5) node  (node12) {};
  \draw (6,6) node (node13) {};
  \draw (8,6) node [dot] (node1111) {};
  \draw (8,7) node [dot, draw=black] (node1411) {};
  \draw (6,7) node  (node1511) {};
  \draw (6,8) node (node1611) {};
   \draw (8,8) node [fulldot,violet] (node14) {};

  \draw (10,5) node (node911) {};
  \draw (10,6) node (node1011) {};
  \draw (10,6) node  (node1211) {};
  \draw (10,7) node (node1311) {};

  \draw (12,8) node [dot] (node15) {};
  \draw (14,8) node  (node16) {};
  \draw (14,9) node  (node17) {};
  \draw (12,9) node [dot, draw=black] (node18) {};
  \draw (10,9) node (node19) {};
  \draw (10,10) node  (node20) {};
  \draw (12,10) node [dot] (node2111) {};
  \draw (14,10) node (node1911) {};
  \draw (14,11) node  (node2011) {};

 \draw (12,11) node [dot, draw=black] (node3111) {};
 
 \draw (10,11) node  (node3112) {};
 
 \draw (10,12) node  (node3113) {};

  \draw (12,12) node [fulldot, violet] (node21) {};
  \draw (16,12) node [dot] (node22) {};
  \draw (18,12) node (node23) {};
  \draw (18,13) node  (node24) {};
  \draw (16,13) node [dot, draw=black] (node25) {};
  \draw (14,13) node  (node26) {};
  \draw (14,14) node (node27) {};
  \draw (16,14) node [dot, draw=black] (node28) {};

    \draw (18,14) node  (node2900) {};
  \draw (18,15) node (node2911) {};

  \draw (16,15) node [dot, draw=black] (node30) {};
  \draw (16,16) node [dot, draw=black] (node31) {};
  
  \draw (14,15) node (node3000) {};
  \draw (14,16) node  (node3100) {};
  
  \draw (18,16) node (node3133) {};
    \draw [grayfactor] (node31) -- (node3133);

  \draw [grayfactor] (node0) -- (node1);
  \draw [fullfactor] (node1) -- (node2.center) -- node [below right=0.5mm] {\footnotesize $\pmb{\alpha}$}
                           (node3.center) -- (node4); 
                           
  \draw [fullfactor] (node111) -- (node211.center) -- node [below right=0.5mm] {\footnotesize $\pmb{\beta}$}
                           (node311.center) -- (node411);

  \draw [grayfactor] (node4) -- (node5.center) -- (node6.center) -- (node111);

  \draw [grayfactor] (node411) -- (node511.center) -- (node611.center) -- (node7);
  \draw [fullfactor, violet] (node7) -- node [above left=.7mm] {\footnotesize $\pmb{\gamma}$} (node8);

  \draw [fullfactor, violet] (node8) -- (node9.center) -- node [below right=0.5mm] {\footnotesize $\pmb{\alpha}$}
                           (node10.center) -- (node11);
  \draw [fullfactor,violet] (node11) -- (node12.center) -- node [above left=0.5mm] {\footnotesize $\pmb{\kappa}$}
                           (node13.center) -- (node1111);
                           
                             \draw [fullfactor,violet] (node1111) -- (node1211.center) -- node [below right=0.5mm] {\footnotesize $\pmb{\beta}$}
                           (node1311.center) -- (node1411);
                           
                           \draw [fullfactor,violet] (node1411) -- (node1511.center) -- node [above left=0.5mm] {\footnotesize $\pmb{\zeta}$}
                           (node1611.center) -- (node14);

  \draw [fullfactor,violet] (node14) -- node [above=.7mm] {\footnotesize $\pmb{\gamma}$} (node15);

  \draw [fullfactor,violet] (node15) -- (node16.center) -- node [below right=0.5mm] {\footnotesize $\pmb{\alpha}$}
                           (node17.center) -- (node18);
  \draw [fullfactor,violet] (node18) -- (node19.center) -- node [above left=0.5mm] {\footnotesize $\pmb{\kappa}$}
                           (node20.center) -- (node2111);
                           
                            \draw [fullfactor,violet] (node2111) -- (node1911.center) -- node [above right=0.5mm]
                             {\footnotesize $\pmb{\beta}$}
                           (node2011.center) -- (node3111);

\draw [fullfactor,violet] (node3111) -- (node3112.center) -- node [above left=0.5mm] {\footnotesize $\pmb{\zeta}$}
                           (node3113.center) -- (node21);
                           
                          \draw [fullfactor] (node30) -- (node3000.center) -- node [above left=0.5mm] {\footnotesize $\pmb{\zeta}$}
                           (node3100.center) -- (node31);

  \draw [fullfactor,violet] (node21) -- node [above=.7mm] {\footnotesize $\pmb{\gamma}$} (node22);
  \draw [grayfactor] (node22) -- (node23.center) -- (node24.center) -- (node25); 
  \draw [fullfactor, black] (node25) -- (node26.center) -- node [above left=0.5mm] {\footnotesize $\pmb{\kappa}$}
                           (node27.center) -- (node28);
\draw [grayfactor] (node28) -- (node2900.center) -- (node2911.center) -- (node30);

\end{scope}
\end{tikzpicture}
}
\caption{Illustration for Lemma \ref{lem:idem}. }\label{fig:idempotent-loop-appendix}
\end{figure}

%% file: appendix-main.tex
\section{Proof of Theorem~\ref{thm:oneway-def}.} \label{ap:main}

Recall that the implication $4\rightarrow 1$ is straightforward,
and the implication $1\rightarrow2$ was already proven in full detail
in the main body. Below, we provide detailed proofs of the implications
$2\rightarrow 3\rightarrow 4$.

\medskip
The implication $2\rightarrow3$ is shown by contradiction. 
Consider a successful run $\r$ of $T$ on some input $u$
and suppose there is an inversion: $\r$ has disjoint loops $I<I'$, 
whose flows contain productive straight edges, say $e$ in $\flow_\r(I)$ 
and $e'$ in $\flow_\r(I')$, such that $e'$ precedes $e$ in the run order. 
Let $u=u_1\, w\, u_2 \, w' \, u_3$ so that $w$ and $w'$ are the
factors of the input delimited by the loops $I$ and $I'$, respectively.
Further let $v$ and $v'$ be the outputs produced along the edges 
$e$ and $e'$, respectively. 
Consider now the run $\r_k$ obtained from $\r$ by pumping 
the input an arbitrary number $k$ of times on the loops $I$ and $I'$. 
This run is over the input $u_1 \, (w)^{k} \, u_2 \, (w')^{k} \, u_3$, 
and in the output produced by $\r_k$ there are $k$ (possibly non-consecutive) 
occurrences of $v$ and $v'$.
By Lemma~\ref{lem:order-preservation} all occurrences of $v'$ precede all
occurrences of $v$. In particular, if $X_1$ (resp. $X_2$) is the set of 
positions corresponding to all the occurrences of $v$ (resp. $v'$) 
in the output produced by $\r_k$, then $(X_1,X_2)$ is a cross of 
width at least $k$.

\medskip
Now we prove the implication $3\rightarrow 4$.
We assume that no run of $T$ has any inversion. 
We want to build a 
partially bijective, regular resynchronizer 
$\Rr$ that is $T$-preserving and such that $\Rr(T)$ is order-preserving.
The resynchronizer $\Rr$ will use input and output parameters to
guess a successful run $\r$ of $T$ on the input $u$ and a corresponding 
\emph{factorization tree} for $\r$ of height at most $H=3|M_T|$ 
(see page~\ref{p:factorization} for the formal definition and
the existence of a factorization tree).

The resynchronizer $\Rr$ that we will define is \emph{functional},
which means here that every source origin is mapped by each
$\move_\t$ formula to at most one target position.


\paragraph*{Notations.}

For a node $\nd$ of a factorization tree we write $I(\nd)$ for 
the input interval which is the yield of the subtree of $\nd$. 
Recall that the leaves of the factorization tree correspond to 
singleton intervals on the input. 
The set of output positions with origins in $I(\nd)$ is denoted by $\out{\nd}$
(note that this might not be an interval).

Recall that an output block $B$ of $\out{\nd}$ 
is a maximal interval of output positions with origins in $I(\nd)$,
and hence the position just before and the position 
just after $B$ have origins outside $I(\nd)$.
We also write $\bout{\nd}$, instead of $\bout{I(\nd)}$, for
the dominant output interval of $I(\nd)$ 
(see page~\pageref{p:bigout} for the definition).
Finally, given a position $x$ in the output 
and a level $\ell$ of the factorization tree of $\r$, 
we denote by $\nd_{x,\ell}$ the unique
node at level $\ell$ such that $I(\nd_{x,\ell})$ 
contains the source origin of $x$. 


\paragraph*{Input Parameters.}
The successful run $\r$ together with its factorization tree of height at most $H=3|M_T|$
can be easily encoded over the input using input parameters $\ipar$.
The parameters describe each input interval $I(\nd)$ and the label 
$\flow(I(\nd))$ of each node $\nd$ in the factorization tree.
Formally, an input interval $I(\nd)$ is described by marking the 
begin and end with two distinguished parameters for the specific level. 
The label $\flow(I(\nd))$ annotates every position inside $I(\nd)$.
This accounts for $H(2+|M_T|)$ input parameters. 
Correctness of the annotations with the above input parameters 
can be expressed by a formula $\ipar$.
In particular, on the leaves, $\ipar$ checks that every interval 
is a singleton of the form $\{y\}$ and its flow is the one induced 
by the letter $u(y)$.
On the internal nodes, $\ipar$ checks that the label of a node
coincides with the monoid product of the labels of its children, which
is a composition of flows. It also checks that for every node with
more than two children, the node and the children are labelled by 
the same idempotent flow.

\paragraph*{Output Parameters.}
We also need to encode the run $\r$ on the output,
because the resynchronizer will determine the target origin
of an output position, not only on the basis of the flow at 
the source origin, but also on the basis of the productive 
transition that generated that particular position.
The annotation that encodes the run $\r$ on the output 
is done using output parameters (one for each transition in $\Delta$),
and its correctness will be enforced by a suitable combination of the 
formulas $\opar$, $\move_\otype$, and $\nxt_{\otype,\otype'}$.
This will take a significant amount of technical details and
will rely on specific properties of formulas $\move_\otype$,
so we prefer to temporarily postpone those details. 

Below, we explain how the origins are transformed by a series
of partial resynchronizers $\Rr_\ell$ that ``converge'' in finitely
many steps to a desired resynchronization, under the assumption
that the output annotation correctly encodes the same run $\r$ that is
represented in the input annotation.

\paragraph*{Moving origins.}

Here we will work with a fixed successful run $\r$ and a factorization tree for it,
that we assume are correctly encoded by the input and output annotations.
For every level $\ell$ of the factorization tree, we will define a
functional,  
bounded, regular resynchronizer $\Rr_\ell$. Each resynchronizer $\Rr_\ell$
will be \emph{partial}, in the sense that for some output positions it will 
not define source-target origin pairs. However, the set of output positions 
with associated source-target origin pairs increases with the level $\ell$,
and the top level resynchronizer $\Rr_*$ will specify source-target 
origin pairs for all output positions. The latter resynchronizer will 
almost define the resynchronization that is needed to prove item (4) of the theorem;
we will only need to modify it slightly in order to make it $1$-bounded
and to check that the output annotation is correct.

To enable the inductive construction, we need the resynchronizer $\Rr_\ell$ 
to satisfy the following properties, for every level $\ell$ of the 
factorization tree:
\begin{itemize}
  \item the set of output positions for which $\Rr_\ell$ defines target origins 
        is the union of the dominant output intervals $\bout{\nd}$ of 
        all nodes $\nd$ at level $\ell$;
  \item $\Rr_\ell$ only moves origins within the same interval at level $\ell$,
        that is, $\Rr_\ell$ defines only pairs $(y,z)$ of source-target origins 
        such that $y,z\in I(\nd)$ for some node $\nd$ at level $\ell$;
  \item the target origins defined by $\Rr_\ell$ are order-preserving
        within the same interval at level $\ell$, that is, 
        for all output positions $x<x'$, if $\Rr_\ell$ defines 
        the target origins of $x,x'$ to be $z,z'$, 
        respectively, and if $z,z'\in I(\nd)$ for some node $\nd$ at level $\ell$,
        then $z\le z'$.
  \item $\Rr_\ell$ is $\ell\cdot4K\obound$-bounded, namely, there are at most
        $\ell\cdot4K\obound$ distinct source origins that are moved by $\Rr_\ell$
        to the same target origin.
\end{itemize}
The inductive construction of $\Rr_\ell$ will basically amount to defining 
appropriate formulas $\move_\otype(y,z)$. 

\smallskip
\noindent{\bf {Base Case}}. 
The base case is $\ell=0$, namely, when the resynchronization is acting
at the leaves of the factorization tree. 
In this case, the regular resynchronizer $\Rr_\ell$ is vacuous, as the 
input intervals $I(\nd)$ associated with the leaves $\nd$ are singletons, 
and hence all dominant output intervals $\bout{\nd}$ are empty.
Formally, for this resynchronizer $\Rr_\ell$, we simply let $\move_\otype(y,z)$ 
be false, independently of the underlying output type $\otype$ and of the 
source and target origins.
This resynchronization is clearly functional,
$0$-bounded,
and order-preserving.

\smallskip
\noindent{\bf {Inductive Step}}. 
For the inductive step, we explain how the origins of an 
output position $x\in\bout{\nd}$ are moved within the interval $I(\nd)$,
where $\nd=\nd_{x,\ell}$ is the node at level $\ell$ that ``generates'' $x$. 
Even though we explain this by mentioning the node $\nd_{x,\ell}$, 
the definition of the resynchronization will not depend
on it, but only on the level $\ell$ and the underlying input 
and output parameters.
In particular, to describe how the origin of a $\otype$-labeled 
output position $x$ is moved, the formula $\move_\otype(y,z)$ 
has to determine the productive edge that generated $x$ 
in the flow that labels the node $\nd_{x,\ell}$.
This can be done by first determining from the output type $\otype$
the productive transition $t_x$ that generated $x$, and then
inspecting the annotation at the source origin $y$ to
``track'' $t_x$ inside the productive edges of the flow 
$\flow(I_{\nd'})$, for each node $\nd'$ along the unique 
path from the leaf $\nd_{x,0}$ to node $\nd_{x,\ell}$.
In the case distinction below, we implicitly rely on this 
type of computation, which can be easily implemented in MSO.
\begin{enumerate}
	\item {\bf{$\pmb{\nd_{x,\ell}}$ is a binary node}}. 
	We first consider the case where $\nd=\nd_{x,\ell}$ is a binary node
	(the annotation on the source origin $y$ will tell us whether this is the case).
	Let $\nd_1,\nd_2$ be the left and right children of $\nd$.
	If $x$ belongs to one of the dominant output blocks $\bout{\nd_1}$ 
	and $\bout{\nd_2}$ (again, this information is available at
	the input annotation), then the resynchronizer $\Rr_\ell$ 
	will inherit the source-target origin pairs associated 
	with $x$ from the lower level resynchronization $\Rr_{\ell-1}$.
	Note that $\bout{\nd_1} < \bout{\nd_2}$ by
	Lemma~\ref{lem:order}, so $\Rr_\ell$ is order-preserving at least 
	for the output positions inside $\bout{\nd_1} \cup \bout{\nd_2}$.
	
	We now describe the source-target origin pairs when 
	$x\in\bout{\nd} \setminus (\bout{\nd_1} \cup \bout{\nd_2})$.
	The idea is to move the origin of $x$
	to one of the following three input positions, depending on the 
	relative order between $x$ and the positions in $\bout{\nd_1}$ and in $\bout{\nd_2}$: 
	\begin{itemize}
		\item the first position of $I(\nd_1)$, if $x < \bout{\nd_1}$;
		\item the last position of $I(\nd_1)$, if $\bout{\nd_1} < x < \bout{\nd_2}$;
		\item the last position of $I(\nd_2)$, if $x > \bout{\nd_2}$.
	\end{itemize}
	Which of the above cases holds can be determined, again, by inspecting
	the output type $\otype$ and the annotation of the source origin $y$,
	in a way similar to the computation of the productive edge that generated
	$x$ at level $\ell$.
	So the described resynchronization can be implemented by an MSO formula
	$\move_\otype(y,z)$.
	
	The resulting resynchronization $\Rr_\ell$ is functional and 
	order-preserving inside every interval at level $\ell$.
 	It remains to argue that $\Rr_\ell$ is $\ell\cdot4K\obound$-bounded. 
	To see why this holds, assume, by the inductive hypothesis, that 
	$\Rr_{\ell-1}$ is $(\ell-1)\cdot4K\obound$-bounded. 
	Recall that the new source-target origin pairs that
	are added to $\Rr_\ell$ are those associated with output positions 
	in $\bout{\nd} \setminus (\bout{\nd_1}\cup\bout{\nd_2})$. 
	Lemma \ref{lem:bin} tells us that there are at most 
	$4K\obound$ distinct positions that are source origins of 
	such positions.
	So, in the worst case, at most $(\ell-1)\cdot4K\obound$ 
	source origins from $\Rr_{\ell-1}$ and
	at most $4K\obound$ new source origins from $\Rr_\ell$
	are moved to the same target origin. 
	This shows that $\Rr_\ell$ is $\ell\cdot4K\obound$-bounded.

	\item {\bf{$\pmb{\nd_{x,\ell}}$ is an idempotent node}}. 
	The case where $\nd=\nd_{x,\ell}$ is an idempotent node with children 
	$\nd_1,\nd_2,\dots,\nd_n$ follows a similar approach.
	For brevity, let $I_i = I(\nd_i)$ and $B_i = \bout{\nd_i}$. 
	By Lemma \ref{lem:order}, we have $B_1 < B_2 < \dots < B_n$.
	Lemma~\ref{lem:idem} then provides a decomposition of $\bout{\nd}$ as 
	$B_1 \juxt J_1 \juxt B_2 \juxt J_2 \juxt \dots \juxt J_{n-1} \juxt B_n$, 
	for some $2K\obound$-small output intervals $J_k$, for $k=1,\dots,n-1$, 
	that have origins inside $I_k\cup I_{k+1}$.

	As before, the resynchronizer $\Rr_\ell$ behaves exactly as $\Rr_{\ell-1}$
	for the output positions inside the $B_k$'s. 
	For any other output position, 
	say $x\in J_k$ for some $k=1,2,\dots,n-1$, 
	we first recall that the source origin $y$ of $x$ 
	is either inside $I_k$ or inside $I_{k+1}$. 
	Depending on which of the two intervals contains $y$, 
	the resynchronizer $\Rr_\ell$ will define the target origin $z$ to 
	be either the last position of $I_k$ or the first position of $I_{k+1}$.
	However, since we cannot determine using MSO the index $k$ of the interval 
	$J_k$ that contains $x$, we proceed as follows.

    First observe that any block $B_i$ can be identified 
    by some flow edge at level $\ell-1$, and the latter edge can
    represented in MSO by suitable monadic predicates over the input.
    Let $B,B'$ be the two consecutive blocks among $B_1,\dots,B_n$ 
    such that $B<x<B'$. Note that these blocks can be determined in MSO
    once the productive edge that generated $x$ is identified.
	Further let $I$ be the interval among $I_1,\dots,I_n$ that 
	contains the origin $y$ of $x$. 
	By the previous arguments, we have that the interval $I$ 
	contains either all the origins of $B$ or all the origins $B'$.
	Moreover, which of the two sub-cases holds can again be determined 
	in MSO by inspecting the annotations. 
	The formula $\move_\otype(y,z)$ can thus define the target origin $z$ to be
	\begin{itemize}
	  \item the last position of $I$, if $I$ contains the origins of $B$;
	  \item the first position of $I$, if $I$ contains the origins of $B'$.
    \end{itemize}

    The above construction yields  a functional
    regular resynchronization $\Rr_\ell$
	that associates with any two output positions $x<x'$ with source origins in the same
	interval $I(\nd)$, some target origins $z\le z'$. In other words, the resynchronization 
	$\Rr_\ell$ is order-preserving in each interval at level $\ell$.
	
	It remains to show that $\Rr_\ell$ is $\ell\cdot4K\obound$-bounded, 
	under the inductive hypothesis that $\Rr_{\ell-1}$ is $(\ell-1)\cdot 4K\obound$-bounded. 
	This is done using a similar argument as before, 
	that is, by observing that the output positions in 
	$\bout{\nd}\setminus\big(\bigcup_{1\le k\le n}\bout{\nd_i}\big)$ 
	belong to some $J_k$, and in the worst case all source origins $y$ 
	of positions from $J_k$ are moved to the last position of $I_k$.
	By Lemma \ref{lem:idem}, there are at most $2K\obound$ such positions $y$.
\end{enumerate}

\smallskip
\paragraph*{Top level resynchronizer.}

Let $\Rr_*$ be the the resynchronizer $\Rr_\ell$ obtained at the top level $\ell$ of the factorization tree.
Based on the above constructions, $\Rr_*$ defines target origins for all output positions,
unless the dominant output interval $\bout{\nd}$ associated with the root $\nd$ is empty
(this can indeed happen when the number of different origins in the output is at most $\obound$,
so not sufficient for having at least one $\obound$-large output factor). 
In particular, if $\bout{\nd}\neq\emptyset$, then $\bout{\nd}$ is the whole output, 
and $\Rr_\ell$ is basically the desired resynchronization, assuming that the output 
annotations are correct.

Let us now discuss briefly the degenerate case where $\bout{\nd}=\emptyset$,
which of course can be detected in MSO.
In this case, the appropriate resynchronizer $\Rr_*$ should be 
redefined so that it moves all source origins to the same target origin, 
say the first input position. 
Clearly, this gives a functional,
regular resynchronizer 
that is order-preserving and $\obound$-bounded.

\paragraph*{Correctness of output annotation.}
Recall that the properties of the top level resynchronizer $\Rr_*$,
in particular, the claim that $\Rr_*$ is bounded,
were crucially relying on the assumption that every output position $x$
is correctly annotated with the productive transition that generated it.
This assumption cannot be guaranteed by the MSO sentence $\opar$ alone
(the property intrinsically talks about a relation between input and
output annotations).
Below, we explain how to check correctness of the output 
annotation with the additional help of the formulas $\move_\otype(y,z)$ 
(that will be modified for this purpose) and $\nxt_{\otype,\otype'}(z,z')$.

Let $\r$ be the successful run as encoded by the input annotation.
The idea is to check that the sequence of productive transitions $t_x$
that annotate the positions $x$ in the output is the maximal sub-sequence 
of $\r$ consisting only of productive transitions. Besides the straightforward
conditions (concerning, for instance, the first and last productive transitions 
of $\r$, or the possible multiple symbols that could be produced within a single
transition), the important condition to be verified is the following:
\begin{align*}
\parbox{0.89\textwidth}{
\emph{For every pair of consecutive output positions $x,x+1$
      with source origins $y,y'$, respectively,
      then on the run $\r$ that is annotated on the input, 
      one can move from transition $t_x$ at position $y$ to 
      transition $t_{x+1}$ at position $y'$ by using at the
      intermediate steps only non-productive transitions.}
}
\tag{$\dagger$}
\end{align*}

The above property is easily expressible by an MSO formula 
$\varphi^\dagger_{\otype,\otype'}(y,y')$, assuming that 
$\otype,\otype'$ are the output types of $x,x+1$ and
the free variables $y$ and $y'$ are interpreted by the 
\emph{source origins} of $x$ and $x+1$, with $x$ ranging 
over all output positions.
This is very close to the type of constraints that can be 
enforced by the formula $\nxt_{\otype,\otype'}$ of a regular 
resynchronizer, with the only difference that the latter formula
can only access the \emph{target origins} $z,z'$ of $x,x+1$.

We thus need a way to uniquely determine from the target origins 
$z,z'$ of $x$ the source origins $y,y'$ of $x$. 
For this, we could rely on the formulas $\move_\otype(y,z)$,
if only they were defining partial bijections between $y$ and $z$.
Those formulas are in fact close to define partial bijections, 
as they are functional and $k$-bounded, for $k=H\cdot4K\obound$.
The latter boundedness property, however, depends again on
the assumption that the output annotation is correct.
We overcome this problem by gradually modifying the resynchronizer 
$\Rr_*$ so as to make it functional and $1$-bounded (i.e., partially bijective), independently
of the output annotations.

We start by modifying the formulas $\move_\otype(y,z)$
to make them ``syntactically'' $k$-bounded.
Formally, we construct from $\move_\otype(y,z)$ the formula
\begin{align*}
  \move'_\otype(y,z) 
  &~=~ \move_\otype(y,z) \\
  &\:~\wedge~\:
       \forall y_1,\dots,y_k,y_{k+1} ~ 
       \Big(\bigwedge\nolimits_i \move_\otype(y_i,z)\Big) 
       \rightarrow 
       \Big(\bigvee\nolimits_{i\neq j} y_i=y_j\Big).
\end{align*}
Intuitively, the above formula is semantically equivalent to 
$\move_\otype(y,z)$ when there are at most $k$  
input positions $y'$ that can be paired with $z$ via the
same formula $\move_\otype$, and it is false otherwise.

Let $\Rr'_*$ be the regular resynchronizer obtained from $\Rr_*$
by replacing the formulas $\move_\otype$ by $\move'_\otype$,
for every output type $\otype$.
By construction, $\Rr'_*$ is 
functional and $k$-bounded, 
independently of any assumption on the output annotations.
We can then apply Lemma \ref{lem:one-bounded} and 
obtain from $\Rr'_*$ an equivalent regular resynchronizer
$\Rr''_*=(\bar I'',\bar O'',\ipar'',\opar'',(\move''_\otype)_\otype,(\nxt''_{\otype,\otype'})_{\otype,\otype'})$ 
that is  $1$-bounded. So each $\move''_\otype$ is a partial bijection.

We are now ready to verify the correctness of the output annotation.
Recall that the idea is to enforce the property $(\dagger)$ 
by exploiting the previously defined formula $\varphi^\dagger_{\otype,\otype'}(y,y')$
and the partial bijection between the source origings $y,y'$ and the
target origins $z,z'$, as defined by $\move''_\otype(y,z)$ and $\move''_{\otype'}(y',z')$.
Formally, we define 
\[
  \nxt'''_{\otype,\otype'}(z,z')
  ~=~ \nxt''_{\otype,\otype'}(z,z') ~\wedge~ 
      \exists y,y' ~ \move''_\otype(y,z) ~\wedge~ \move''_{\otype'}(y',z') ~\wedge~
      \varphi^\dagger_{\otype,\otype'}(y,y').
\]
To conclude, by replacing in $\Rr''$ the formulas $\nxt''_{\otype,\otype'}$ 
with $\nxt'''_{\otype,\otype'}$, we obtain a regular resynchronizer $\Rr$ 
that is 
partially bijective, $T$-preserving and such that $\Rr(T)$ is
order-preserving. This completes the proof of the implication $3\rightarrow 4$
of our main theorem.
\qed


%% file: appendix-unbounded.tex
\section{Proof of Theorem \ref{thm:decidability}}\label{app:general-case}

We provide here the missing details of the proof of Theorem \ref{thm:decidability},
as sketched in Section \ref{sec:general-case}.
We recall that the goal is to construct, from a given arbitrary two-way 
transducer $T$:
\begin{enumerate}
  \item a bounded-visit transducer $\short T$ that is classically equivalent to $T$, 
  \item 
    partially bijective, regular resynchronizer $\Rr$ that is $T$-preserving
        and such that $\Rr(T) \oeq \short T$.
\end{enumerate}

We will reason with a fixed input $u$ at hand and with an induced 
accessibility relation on productive transitions of $T$, tagged with origins.
Formally, a \emph{tagged transition} is any pair $(t,y)$ consisting of a 
transition $t\in\Delta$ and a position $y$ on the input $u$, such that $t$ 
occurs at position $y$ in some successful run on $u$.
The accessibility preorder on tagged transitions is such that $(t,y) \preceq_u (t',y')$
whenever $T$ has a run on $u$ starting with transition $t$ at position $y$
and ending with transition $t'$ at position $y'$.
This preorder induces an equivalence relation, denoted $\sim_u$.
Intuitively, $(t,y)\sim_u (t',y')$ means that $T$ can cycle an arbitrary
number of times between these two tagged transitions
(possibly $(t,y)=(t',y')$). A $\sim_u$-equivalence class $C$ is called
\emph{realizable on $u$} if there is a successful run on $u$ that uses
at least once a tagged transition from the class $C$.
  
We say that $T$ is \emph{$K$-sparse} if for every input $u$
and every realizable $\sim_u$-equivalence class $C$, there are at most
$K$ \emph{productive} tagged transitions in $C$
(recall that a productive transition is one that produces non-empty output).  
Intuitively, bounded sparsity means that the number of origins of outputs 
produced by vertical loops in successful runs of $T$ is uniformly bounded. 
If $T$ is not $K$-sparse for any $K$, then we say that $T$ has \emph{unbounded sparsity}.

When $T$ is $K$-sparse, the productive tagged transitions
from the same realizable $\sim_u$-equivalence class can be lexicographically 
ordered and distinguished by means of numbers from a fixed finite range, 
say $\{1,\dots,K\}$.
An important observation is that the equivalence $\sim_u$ is a 
regular property, in the sense that one can construct, for instance, 
an MSO formula $\varphi^{\sim_u}_{t,t'}(y,y')$ that holds on input $u$
if and only if $(t,y)\sim_u (t',y')$.
In particular, this implies that unbounded sparsity can be effectively
tested: it suffices to construct the regular language consisting of
every possible input $u$ with a distinguished realizable $\sim_u$-equivalence 
class marked on it, and check whether this language contains words with 
arbitrarily many marked positions that correspond to productive tagged
transitions (this boils down to detecting special 
loops in a classical finite-state automaton).

\begin{lemma}\label{lem:unbounded-sparsity}
If $T$ has unbounded sparsity, then $T$ is not one-way definable.
\end{lemma}

\begin{proof}
The assumption that $T$ has unbounded sparsity and the definition of $\sim_u$ 
imply that, for every $n\in\mathbb{N}$, there exist an input $u$, 
a successful run $\rho$ on $u$, and $2n$ tagged transitions 
$(t_1,y_1),\dots,(t_n,y_n)$, $(t'_1,y'_1),\dots,(t'_n,y'_n)$
such that the $t_i$'s occur before the $t'_j$ in $\rho$ and
the $y'_i$ are to the right of the $y'_j$.
Since $n$ can grow arbitrarily, this witnesses precisely the fact that $
T$ has unbounded cross-width.
Thus, by the implication $1\rightarrow 2$ of Theorem~\ref{thm:oneway-def},
which is independent of $T$ being bounded-visit, we know that $T$ is 
not one-way resynchronizable.
\qed
\end{proof}

\medskip
Let us now show how to construct a bounded-visit transducer $\short T$ 
with regular outputs and common guess that is equivalent to $T$, under the
assumption that $T$ is $K$-sparse for some constant $K$.
Intuitively, $\short T$ simulates successful runs of $T$ on input $u$
by shortcutting maximal vertical loops. 
Formally, for an input $u$ and a tagged transition $(t,y)$,
a \emph{vertical loop} at $(t,y)$ is any run on $u$ that starts and ends 
with transition $t$ at position $y$. 
We will tacitly focus on vertical loops that are realizable on $u$,
exactly as we did for $\sim_u$-equivalence classes.
The output of a vertical loop is the word spelled out by the 
productive transitions in it.

Of course, all tagged transitions in a vertical loop at $(t,y)$ 
are $\sim_u$-equivalent to $(t,y)$. In particular, as $T$ is $K$-sparse, 
there are at most $K$ productive tagged transitions in a
(realizable) vertical loop, 
and hence the language $L_{t,y}$ of outputs of vertical loops at $(t,y)$ 
is regular. In addition, there are only finitely many languages
$L_{t,y}$ for varying $(t,y)$.  This can be seen as follows: we can assume an order on the
elements of the $\sim_u$-class $C$ of $(t,y)$, and a strongly connected
graph with nodes corresponding to $C$ and edges reflecting the accessibility
preorder. The correctness of the graph can be checked with regular
annotations on the input, and the graph itself can be turned into an
automaton accepting $L_{t,y}$. 
Therefore, using common guess in $\short T$, we can assume 
that every position $y$ carries as annotation the language $L_{t,y}$ 
for each transition $t$. By definition, $L_{t,y}$ is non-empty if and only
if there is some productive vertical loop at $(t,y)$. 

Consider an arbitrary successful run $\rho$ of $T$ on $u$.
Let $\short\rho$ be the run obtained by replacing, from left to right,
every maximal vertical loop at $(t,y)$ by the single transition
$t$. Here, maximality refers to the subrun relation. 
We call $\short\rho$ the \emph{normalization of $\rho$} and we observe that
this is a successful, $|\D|$-visit run.
This means that (i) $\short\rho$ can be finitely encoded on the input as a sequence 
of flows of height at most $|\Delta|$, and (ii) the language consisting of inputs 
annotated with such encodings is regular.

The transducer $\short T$ guesses the encoding of a normalization $\short\rho$ 
and uses it to simulate a possible run $\rho$ of $T$. 
In particular, every time $\short T$ traverses a transition $t$ from 
the flow of $\short\rho$ at position $y$, it outputs a word from the language $L_{t,y}$.
However, in order to simplify later the construction of a resynchronizer $\Rr$
such that $\Rr(T) \oeq \short T$, 
it is convenient that $\short T$ outputs the word 
from $L_{t,y}$ in a possibly different origin, which is uniquely determined
by the $\sim_u$-equivalence class of $(t,y)$. Formally, we define the
\emph{anchor} of a $\sim_u$-equivalence class $C$, denoted $\an{C}$, 
to be the leftmost input position $z$ such that $(t',z)\in C$ for some transition $t'$. 
After traversing a transition $t$ from the flow at position $y$, and before
outputting a word from $L_{t,y}$, the transducer $\short T$ moves to the
anchor $\an{[(t,y)]_{\sim_u}}$. Then it outputs the appropriate word and
moves back to position $y$, where it can resume the simulation of 
the normalized run $\short\rho$. 
Note that the position $y$ can be recovered from the anchor $\an{[(t,y)]_{\sim_u}}$
since the elements inside the equivalence class $[(t,y)]_{\sim_u}$
can be identified by numbers from $\{1,\dots,K\}$ (recall that $T$ is $K$-sparse),
and since the relationship between any two such elements is a regular property.
It is routine to verify that the described transducer $\short T$ 
is equivalent to $T$ and bounded-visit.

\medskip
Let us now explain how to construct a 
partially bijective,
regular resynchronizer $\Rr$ that is $T$-preserving and 
such that $\Rr(T) \oeq \short T$.
We proceed as in the construction of $\short T$ by annotating the
input word $u$
with flows that encode the normalization $\short\rho$ of a successful
run $\rho$ of $T$ on $u$.
As for the output word $v$, we annotate every position $x$  of $v$
with the productive transition $t=(q,a,v,q')$ of $\rho$ that generated $x$.
For short, we call $t$ \emph{the transition of $x$}.
In addition, we fix an MSO-definable total ordering on tagged transitions
(e.g.~the lexicographic ordering). Then, we determine from each output 
position $x$ the $\sim_u$-equivalence class $C=[(t,y)]_{\sim_u}$, where 
$u$ is the underlying input, $t$ is the productive transition that generated $x$, 
and $y$ is its origin, and we extend the annotation of $x$ with the index 
of the element $(t,y)$ inside the equivalence class $C$, according to the 
fixed total ordering on tagged transitions. This number $i$ is called
\emph{the index of $x$}.

The resynchronizer $\Rr$ needs to redirect the source origin 
$y$ of any output position generated by a transition $t$ 
to a target origin $z$ that is the anchor of the $\sim_u$-equivalence 
class of $(t,y)$.
To simplify the explanation, we temporarily assume that the 
input and output are correctly annotated as described above.
By inspecting the type $\otype$ of an output position $x$,
the formula $\move_\otype(y,z)$ of $\Rr$ can determine the 
transition $t$ of $x$, and enforce that 
$(t,y)\sim_u (t',z)$, for some transition $t'$, and 
that $(t,y)\not\sim_u (t'',z')$, for all $z'<z$ and all transitions $t''$.
\emph{Under the assumption that the input and output annotations are correct},
this would result in a bounded resynchronizer $\Rr$. 
Indeed, for every position $z$, there exist at most $K\cdot|\Delta|$ 
positions $y$ that, paired with some productive transition, turn out to be
$\sim_u$-equivalent to $(t',z)$ for some transition $t'$. 
Once again, we need to further constrain the relation $\move_\otype(y,z)$ 
so that it describes a partial bijection between source and target origins
(this will be useful later).
For this, it suffices to additionally enforce that 
$(t,y)$ is the $i$-th element in its $\sim_u$-equivalence class,
accordingly to the fixed total ordering on tagged transitions,
where $i$ is the index specified in the output type $\otype$ of $x$.
This latter modification also guarantees that
$i$ is the correct index of $x$.

Unless we further refine our constructions,
we cannot claim that they always result in a $1$-bounded resynchronizer $\Rr$,
since the above arguments crucially rely on the assumption that 
the input and output annotations are correct. 
However, we can apply the same trick that we
used in the proof of Theorem \ref{thm:oneway-def}, 
to make the resynchronizer $\Rr$ ``syntactically'' $1$-bounded,
even in the presence of badly-formed annotations. 
Formally, let $\move_\otype(y,z)$ be the formula 
that transforms the origins in the way described above, 
and define 
\[
  \move'_\otype(y,z) ~=~ \move_\otype(y,z) ~\wedge~ 
                         \forall y' ~ \big(\move_\otype(y',z) \rightarrow y'=y\big).
\]
By construction, the above formula defines a partial 
bijection entailing the old relation $\move_\otype$
(in the worst case, when the annotations are not correct,
the above formula may not hold for some pairs of source 
and target origins).
In addition, if the annotations are correct, then 
$\move'_\otype(y,z)$ is semantically equivalent to 
$\move_\otype(y,z)$, as desired.
In this way, we obtain a regular resynchronizer 
$\Rr=(\bar I,\bar O,\ipar,\opar,\move'_\otype,\nxt)$
that is always $1$-bounded, no matter how we define 
$\ipar$, $\opar$, and $\nxt$. 

\medskip
We now explain how to check that the annotations are correct.
The input annotation does not pose any particular problem, since
the language of inputs annotated with normalized runs is regular,
and can be checked using the first formula $\ipar$ of the resynchronizer.
As for the output annotation, correctness of the indices 
was already enforced by the $\move'_\otype$ relation.
It remains to enforce correctness of the transitions.
Once again, this boils down to verifying the following property ($\dagger$):
\begin{align*}
\parbox{0.89\textwidth}{
\emph{For every pair of consecutive output positions $x,x+1$
      with source origins $y,y'$, respectively,
      if $t,t'$ are the productive transitions specified in the output
      types of $x,x+1$, 
      then on the flows that annotate the input, one can move from
      transition $t$ at position $y$ to transition $t'$ at position $y'$
      by using as intermediate edges only non-productive transitions.}
}
\tag{$\dagger$}
\end{align*}

\noindent
From here we proceed exactly as in the proof of Theorem \ref{thm:oneway-def}.
We observe that Property ($\dagger$) is expressible by an MSO formula 
$\varphi_{\otype,\otype'}^{\dagger}(y,y')$, assuming that 
$\otype,\otype'$ are the output types of $x,x+1$,
that $y,y'$ are interpreted by the source origins of $x,x+1$, 
and that $x$ ranges over all output positions.
We then recall that $\move_\otype(y,z)$ and $\move_\otype(y',z')$ 
describe partial bijections between source and target origins,
and exploit this enforce ($\dagger$) by means of the last formula
of $\Rr$:
\[
  \nxt_{\otype,\otype'}(z,z')
  ~=~ \exists y,y' ~ \move_\otype(y,z) ~\wedge~ \move_{\otype'}(y',z') ~\wedge~
      \varphi_{\otype,\otype'}^\dagger(y,y').
\]
This guarantees that all annotations are correct, and proves
that $\Rr$ is a 
partially bijective, regular resynchronizer 
satisfying $\Rr(T) \oeq \short T$.
It is also immediate to see that $\Rr$ is $T$-preserving.

\medskip
We finally prove that one-way resynchronizability of $T$ reduces
to one-way resynchronizability of $\short T$, which can be effectively 
tested using Theorem \ref{thm:oneway-def} since $\short T$ is bounded-visit:

\Robustness*

\begin{proof}
For the right-to-left implication, suppose that $T' \oeq \Rr(T)$ is bounded-visit 
and one-way resynchronizable. Since $T'$ is bounded-visit, 
we can use the implications 
$1\rightarrow 2\rightarrow 3\rightarrow 4$ in
Theorem~\ref{thm:oneway-def} to get the existence of a
partially bijective, regular resynchronizer $\Rr'$ that is $T'$-preserving and such that
$\Rr'(T')$ is order-preserving.
By Lemma \ref{lem:resync-closure}, 
there is a bounded, regular resynchronizer $\Rr''$ that
is equivalent to $\Rr'\circ\Rr$. In particular, $\Rr''(T)$ is order-preserving.
It remains to verify that $\Rr''$ is also $T$-preserving. 
Consider any synchronized pair $\s \in\sem{T}_o$. Since $\Rr$ is $T$-preserving,
$\s$ belongs to the domain of $\Rr'$, and hence $(\s,\s')\in\Rr$ for some synchronized
pair $\s'\in\sem{T'}_o$. Since $\Rr$ is $T'$-preserving, $\s'$ belongs to the domain
of $\Rr$, and hence there is $(\s,\s'')\in(\Rr'\circ\Rr) = \Rr''$. 
This shows that $\Rr''$ is $T$-preserving, and hence $T$ is one-way resynchronizable.

\medskip
For the converse direction, 
suppose that 
$T'$ is bounded-visit, but not one-way resynchronizable. 
We apply again Theorem~\ref{thm:oneway-def}, but now we use the contrapositives 
of the implications $2\rightarrow 3\rightarrow 4\rightarrow 1$, and obtain 
that $T'$ has unbounded cross-width 
(see Definition~\ref{def:crosswidth}).

We also recall that 
$\Rr=(\bar I,\bar O,\ipar,(\move_{\otype})_{\otype}, (\nxt_{\otype,\otype'})_{\otype,\otype})$ 
is partially bijective.  
This means that every formula $\move_\otype(y,z)$
defines a partial bijection from source to target positions.
A useful property of every MSO-definable partial bijection 
is that, for every position $t$, it can only define boundedly 
many pairs $(y,z)$ with either $y\le t<z$ or $z\le t<y$
--- for short, we say call such a pair $(y,z)$ \emph{$t$-separated}.
This follows from compositional properties of regular languages.
Indeed, let $\Aa$ be a deterministic automaton equivalent 
to the formula that defines the partial bijection. For every 
pair $(y,z)$ in the partial bijection, let $q_{y,z}$ be the
state visited at position $t$ by the successful run of $\Aa$ 
on the input annotated with the pair $(y,z)$.
If $\Aa$ accepted more than $|Q|$ pairs that are $t$-separated,
where $Q$ is the state space of $\Aa$, then at least two of them,
say $(y,z)$ and $(y',z')$, would satisfy $q_{y,z} = q_{y',z'}$. 
But this would imply that the pair $(y,z')$ is also 
accepted by $\Aa$, which contradicts the assumption that $\Aa$
defines a partial bijection.

We now exploit the above result to prove that the property 
of having unbounded cross-width transfers from $T'$ to $T$.  
Consider a cross $(X_1,X_2)$ of arbitrarily large width $h$
in some synchronized pair $\s=(u,v,\orig)$ of $T'$.
Without loss of generality, assume that all positions in 
$X_1 \cup X_2$ have the same type $\otype$. 
Let $Z_i=\orig(X_i)$, for $i=1,2$, and $t=\max(Z_2)$.
By definition of cross, we have $X_1<X_2$ and $Z_2\le t < Z_1$.
Recall that $\move_\otype$ defines a partial bijection,
and that this implies that there are only boundedly many
pairs of source-target origins that are $t$-separated,
say $(y_1,z_1),\dots,(y_k,z_k)$ for a constant $k$
that only depends on $\Rr$.
Moreover, since $\Rr(T) \oeq T'$, the positions in $Z_i$ can
be seen as target origins of the formula $\move_\otype$ of $\Rr$. 
Now, let $X'_i = X_i \setminus \orig^{-1}(\{z_1,\dots,z_k\}$
and $Y'_i = \orig'(X'_i)$, for any synchronized pair
$\s'=(u,v,\orig')$ such that $(\s,\s')\in\Rr$.
By construction, we have $X'_1<X'_2$ and $Y'_2\le t < Y'_1$
(the latter condition follows from the fact that the source 
origins from $Y'_i$ can only be moved to target 
origins on the same side w.r.t.~$t$).
This means that $(X'_1,X'_2)$ is a cross of width $h-k$.
As $h$ can be taken arbitrarily large and $k$ is constant, 
this proves that $T$ has unbounded cross-width as well.



Finally, by the contrapositive of the implication $1\rightarrow 2$ of 
Theorem~\ref{thm:oneway-def} (which does not need the assumption 
that $T$ is bounded-visit), we conclude that $T$ is not one-way resynchronizable.
\qed
\end{proof}

\medskip
Summing up, the algorithm that decides whether a given two-way transducer $T$ 
is one-way resynchronizable first verifies that $T$ is $K$-sparse for 
some $K$ (if not, it claims that $T$ is not one-way resynchronziable),
then it constructs a bounded-visit transducer $\short T$ equivalent to $T$,
and finally decides whether $\short T$ is one-way resynchronizable 
(which happens if and only if $T$ is one-way resynchronizable).
This concludes the proof of Theorem \ref{thm:decidability}.
\qed
